\documentclass{llncs}
\usepackage{latexsym}
\usepackage{amssymb}
\usepackage{amsmath}
\usepackage{stmaryrd}
\usepackage{array}
\usepackage[all]{xy}

\outer\long\def\BLURB#1{}

\def\real{\Vdash}

\def\limp{\Rightarrow}
\def\liff{\Leftrightarrow}
\def\<{\langle}
\def\>{\rangle}
\def\fa{\forall}
\def\ex{\exists}
\def\ra{\rightarrow}
\def\lra{\longrightarrow}  
\def\la{\leftarrow}
\def\Dom{\mathop{\mathrm{Dom}}\nolimits}
\def\Cod{\mathop{\mathrm{Cod}}\nolimits}

\def\Clos{\mathop{\mathrm{Clos}}\nolimits}
\def\CST{\texttt{cst}}
\def\FUN{\texttt{fun}}
\def\PRE{\texttt{pre}}
\def\P{\mathfrak{P}}
\def\declaresymbols#1#2{\vspace{12pt}\par\noindent
  \begin{tabular}{|>{$}p{20mm}<{$}|>{$}p{25mm}<{$}|p{63mm}|}
   \hline \multicolumn{3}{|l|}{\textbf{#1}} \\\hline #2\hline
  \end{tabular}\par
}

\def\M{\mathcal{M}}

\def\incirc{\mathrel{{\in}^{\circ}}}
\def\memcirc{\mathop{{\mem}^{\circ}}}

\def\Cl{\mathop{\mathrm{Cl}}}
\def\LBrc{{\{\mskip-4.5mu|}}
\def\RBrc{{|\mskip-4.5mu\}}}
\def\Nat{\mathit{Nat}}
\def\Pred{\mathit{Pred}}
\def\Null{\mathit{Null}}
\def\N{\mathbb{N}}
\def\Zermod{\mathsf{IZ}^{\mathrm{mod}}}
\def\Zst{\mathsf{IZ}^{\mathop{\mathrm{st}}}}

\def\ZF{\mathsf{ZF}}

\def\Zclass{\mathsf{IZ}^{\mathop{\mathrm{class}}}}
\def\Zskol{\mathsf{IZ}^{\mathop{\mathrm{skol}}}}
\def\ZskolS{\mathsf{IZ}^{\mathop{\mathrm{skol2}}}}
\def\Reif{\mathop{\mathrm{Reif}}\nolimits}
\def\Collapse{\mathop{\mathrm{Collapse}}\nolimits}
\def\ISeg{\mathop{\mathrm{ISeg}}\nolimits}
\def\Rgraph{\mathop{\mathrm{Rgraph}}\nolimits}
\def\Pgraph{\mathop{\mathrm{Pgraph}}\nolimits}
\def\Sgraph{\mathop{\mathrm{Sgraph}}\nolimits}
\def\Spgraph{\mathop{\mathrm{Spgraph}}\nolimits}
\def\Graph{\mathop{\mathrm{Graph}}\nolimits}
\def\Int#1{\llbracket#1\rrbracket}
\def\root{\mathop{\mathrm{root}}}
\def\rel{\mathop{\mathrm{rel}}}

\def\mem{\mathop{\mathrm{mem}}}
\def\Empty{\mathop{\mathit{Empty}}}
\def\Succ{\mathop{\mathit{Succ}}}
\def\Ind{\mathop{\mathit{Ind}}}
\def\Set{\mathop{\mathrm{Set}}}
\def\Class{\mathop{\mathrm{Class}}}

\def\Car#1{\overline{#1}}

\def\tlor{\mathrel{\tilde{\lor}}}

\def\cast#1{\lfloor#1\rfloor}

\newcommand\couic[1]{}

\newcommand\interp[1]{#1^*}

\newbox\tempa
\newbox\tempb
\newdimen\tempc
\def\mud#1{\hfil $\displaystyle{\mathstrut #1}$\hfil}
\def\rig#1{\hfil $\displaystyle{#1}$}
\def\irulehelp#1#2#3{\setbox\tempa=\hbox{$\displaystyle{\mathstrut #2}$}%
                        \setbox\tempb=\vbox{\halign{##\cr
        \mud{#1}\cr
        \noalign{\vskip\the\lineskip}%
        \noalign{\hrule height 0pt}%
        \rig{\vbox to 0pt{\vss\hbox to 0pt{${\; #3}$\hss}\vss}}\cr
        \noalign{\hrule}%
        \noalign{\vskip\the\lineskip}%
        \mud{\copy\tempa}\cr}}%
                      \tempc=\wd\tempb
                      \advance\tempc by \wd\tempa
                      \divide\tempc by 2 }
\def\irule#1#2#3{{\irulehelp{#1}{#2}{#3}%
                     \hbox to \wd\tempa{\hss \box\tempb \hss}}}
\def\Proof{\mathrm{Proof}}
\def\CR{\mathcal{CR}}
\def\SN{\mathrm{SN}}

\begin{document}

\title{Cut elimination for Zermelo set theory}
\author{Gilles Dowek\inst{1} \and Alexandre Miquel\inst{2}}
\institute{
  {\'E}cole polytechnique and INRIA\\
  LIX, {\'E}cole polytechnique, 91128 Palaiseau Cedex, France \\
  \email{Gilles.Dowek@polytechnique.edu} \and
  Universit{\'e} Paris~7,\\
  PPS, 175 Rue du Chevaleret, 75013 Paris, France \\
  \email{Alexandre.Miquel@pps.jussieu.fr} }
\pagestyle{headings} 

\maketitle

\begin{abstract}
  We show how to express intuitionistic Zermelo set theory in
  deduction modulo (\emph{i.e.}\ by replacing its axioms by rewrite
  rules) in such a way that the corresponding notion of proof enjoys
  the normalization property.
  To do so, we first rephrase set theory as a theory of pointed graphs
  (following a paradigm due to P.~Aczel) by interpreting set-theoretic
  equality as bisimilarity, and show that in this setting, Zermelo's
  axioms can be decomposed into graph-theoretic primitives that can be
  turned into rewrite rules.  We then show that the theory we obtain
  in deduction modulo is a conservative extension of (a minor
  extension of) Zermelo set theory.
  Finally, we prove the normalization of the intuitionistic fragment
  of the theory.
\end{abstract}

The cut elimination theorem is a central result in proof theory that
has many corollaries such as the disjunction property and the witness
property for constructive proofs, the completeness of various proof
search methods and the decidability of some fragments of predicate
logic, as well as some independence results.

However, most of these corollaries hold for pure predicate logic and
do not generally extend when we add axioms, because the property that
cut-free proofs end with an introduction rule does not generalize in
the presence of axioms.
Thus, extensions of the normalization theorem have been proved for
some axiomatic theories, for instance arithmetic, simple type
theory~\cite{Girard70,Girard72} or the so-called stratified
foundations~\cite{Crabbe91}.  There are several ways to extend
normalization to axiomatic theories: the first is to consider a
special form of cut corresponding to a given axiom, typically the
induction axiom.  A second is to transform axioms into deduction
rules, typically the $\beta$-equivalence axiom. A third way is to
replace axioms by computation rules and consider deduction rules
modulo the congruence generated by these computation rules
\cite{DHK,DowekWerner}.

Unfortunately, extending the normalization theorem to set theory has
always appeared to be difficult or even impossible: a counter example,
due to M.~Crabb\'e~\cite{Crabbe74} shows that normalization does not
hold when we replace the axioms of set theory by the obvious deduction
rules, and in particular the Restricted Comprehension axiom by a
deduction rule allowing to deduce the formula $a \in b \land P(x \la
a)$ from $a \in \{x \in b~|~P\}$ and vice-versa. In the same way,
normalization fails if we replace the comprehension axiom by a
computation rule rewriting $a \in \{x \in b~|~P\}$ to $a \in b \land
P(x \la a)$.  Calling $C$ the set $\{x \in A~|~\neg x \in x\}$ the
formula $C \in C$ rewrites to $C \in A \land \neg C \in C$ and it is
not difficult to check that the formula $\neg C \in A$.
This counterexample raises the following question: is the failure of
normalization an artifact of this particular formulation of set
theory, or do all formulations of this theory have a similar property?

More recently, interpretations of set theory in type theory have
been proposed~\cite{MiqPhD,Miq03,Miq04} that follow P.~Aczel's ``sets
as pointed graphs'' paradigm~\cite{Acz88} by interpreting sets as
pointed graphs and extensional equality as bisimilarity. One
remarkable feature about these translations is that they express set
theory in a framework that enjoys normalization. Another is that
although the formul{\ae}
$a \in \{x \in b~|~P\}$ and $a \in b \land P(x \la a)$
are provably equivalent, their proofs are different.
This suggests that the failure of
normalization for set theory is not a property of the theory itself,
but of some particular way to transform the axioms into deduction or
computation rules.

In the type theoretic interpretation of set theory where sets are
translated as pointed graphs, the membership relation~$\in$ is no
longer primitive, but defined in terms of other atomic relations such
as the ternary relation $x~\eta_{a}~y$ expressing that two nodes~$x$
and~$y$ are connected by an edge in a pointed graph~$a$.

In this paper, we aim at building a theory of pointed graphs
in predicate logic---that we call $\Zermod$---which is expressive
enough to encode set theory in a conservative way.
For that, we start from a simple extension of intuitionistic Zermelo
set theory (without foundation) called $\Zst$, namely, Zermelo set
theory with the axioms of Strong Extensionality and Transitive
Closure.

Instead of expressing the theory~$\Zermod$ with axioms, we
shall directly express it with computation rules.
It is well-known~\cite{DHK} that any theory expressed with computation
rules can also be expressed with axioms, replacing every computation
rule of the form $l \lra r$ by the axiom $l=r$ when~$l$ and~$r$ are
terms, or by the axiom $l \liff r$ when~$l$ and~$r$ are formul{\ae}.
Expressing this theory with rewrite rules instead of axioms makes the
normalization theorem harder to prove but is a key element for the
cut-free proofs to end with an introduction rule.
To prove our normalization theorem, we shall use two main
ingredients: reducibility candidates as introduced by J.-Y. Girard
\cite{Girard70} to prove normalization for higher-order logic, and
the forcing/realizability method, following~\cite{Crabbe91,realizmod}.

\section{Deduction modulo}

In deduction modulo, the notions of language, term and formula are
that of first-order predicate logic. But, a theory is formed with a set of
axioms $\Gamma$ \emph{and a congruence $\equiv$} defined on
formul{\ae}. Such a congruence may be defined by a rewrite systems on
terms and on formul{\ae}.  Then, the deduction
rules take this congruence into account. For instance, the
\emph{modus ponens} is not stated as usual
$$\irule{A \limp B~~~A}{B}{}$$
as the first premise need not be exactly $A \limp B$ but may be
only congruent to this formula, hence it is stated
$$\irule{C~~~A}{B}{\mbox{if $C \equiv A \limp B$}}$$
All the rules of natural deduction may be stated in a
similar way. See, for instance, \cite{DowekWerner} for a complete presentation.

For example, arithmetic can be defined by a congruence defined
by the following rewrite rules
$$\begin{array}{r@{~}c@{~}l}
  0 + y &\lra& y \\ S(x) + y &\lra& S(x+y) \\
\end{array}\qquad\begin{array}{r@{~}c@{~}l}
  0 \times y &\lra& 0 \\
  S(x) \times y &\lra& x \times y + y \\
\end{array}$$
and some axioms, including the identity axiom $\fa x~(x = x)$. 
In this theory, we can prove that the number $4$ is even 
$$\irule{\irule{\irule{} 
    {\Gamma \vdash \forall x~x = x}
    {\mbox{axiom}}} 
  {\Gamma \vdash 2 \times 2 = 4}
  {\mbox{$(x,x = x,4)$ $\forall$-elim}}} 
{\Gamma \vdash \exists x~2 \times x = 4} 
{\mbox{$(x,2 \times x = 4,2)$ $\exists$-intro}}\qquad\qquad\qquad$$
Substituting the term $2$ for the variable $x$ in the formula
$2\times x=4$ yields $2 \times 2 = 4$, that is congruent to $4=4$.
The transformation of one formula into the other, that requires
several proof steps in usual formulation of arithmetic, is dropped
from the proof in deduction modulo.

Deduction modulo allows rules rewriting terms to terms, but also 
atomic formul{\ae} to arbitrary ones. For instance
$$x \times y = 0 \lra x = 0 \lor y = 0$$

When we take the rewrite rules above, the axioms of addition 
and multiplication are not needed anymore as, for example, the 
formula $\forall y~0 + y = y$ is congruent to the axiom
$\forall y~y = y$. Thus, rewrite rules replace axioms. 

This equivalence between rewrite rules and axioms is expressed by
the \emph{equivalence lemma}, which says that for every
congruence~$\equiv$ we can find a theory $\mathcal{T}$ such that
$\Gamma\,{\vdash}\,A$ is provable in deduction modulo the congruence
$\equiv$ if and only if $\mathcal{T},\Gamma\,{\vdash}\,A$ is provable
in ordinary first-order predicate logic~\cite{DHK}.
Hence, deduction modulo is not a true extension of predicate logic,
but rather an alternative formulation of predicate logic. Of course,
the provable formul{\ae} are the same in both cases, but the proofs
are very different.

\section{Variations on axiomatic Set Theory}

In this section, we define the theory $\Zst$.
This theory is Zermelo set theory extended with two axioms: the Strong
Extensionality axiom (which replaces the standard Extensionality axiom
of set theory) and the Transitive Closure axiom.
In Zermelo-Fraenkel set theory with the Foundation axiom, the
Strong Extensionality axiom can be derived from the Foundation axiom,
but it is weaker.
Similarly, the Transitive Closure axiom is a consequence of the
Replacement scheme, but it is weaker.

Since the theory $\Zst$ is expressed in the standard existential way,
we also define two conservative extensions of~$\Zst$ plus a
non-conservative extension.
The first extension of~$\Zst$ is a theory called~$\Zclass$ obtained
by adding a conservative notion of class \emph{{\`a} la} Von
Neumann-Bernays-G\"odel.
The second extension, called $\Zskol$, is built from the latter by
adding Skolem symbols to denote sets and class, including notations to
denote sets and class defined by comprehension.
As we shall see in section~\ref{sec:TranslZermodZst}, such a
conservative extension of the language of set theory is convenient
to define the translation which maps formul{\ae} of the
language~$\Zermod$ of pointed graphs back to the language of set
theory.

The final extension, called $\ZskolS$, is an extension of $\Zskol$
with impredicative classes that will be used in
section~\ref{sec:Normalization}.
This extension of $\Zskol$ is nothing but a skolemized presentation of
second order Zermelo set theory with strong extensionality and
transitive closure.

\subsection{The theory $\Zst$}

\begin{definition}[The theory $\Zst$]  The theory $\Zst$ is expressed
  in predicate logic.  Its language is the language of first-order
  predicate logic formed with two binary predicate symbols $=$ and
  $\in$, and its axioms are given in Table~\ref{tbl:ZstAxioms}.
\end{definition}
We use the standard abbreviations:
$$\begin{array}{r@{~~}c@{~~}l}
  a \subseteq b & \equiv & \fa x\ (x \in a \limp x \in b)\\
  \Empty(a) & \equiv & \fa x\ \neg (x \in a)\\
  \Succ(a,b) & \equiv &
  \fa x\ (x \in b \liff (x \in a \lor x=a)) \\
  \Ind(c) &\equiv&
  \fa a\ (\Empty(a) \limp a \in c)  \land 
  \fa a\ (a \in c \limp \fa b~(\Succ(a,b) \limp b \in c))\\
  \Nat(a) &\equiv& \fa b\ (\Ind(b) \limp a \in b) 
\end{array}$$

\begin{table}[tp]
  $\begin{array}{l>{\quad}l}
    \hline\hline\\[-3pt]
    (\text{Reflexivity}) & \fa x\ (x=x) \\[6pt]
    (\text{Equ. Compat.}) & \fa x \fa x' \fa y~
    (x=x' \land  x=y\limp x'=y) \\[6pt]
    (\text{Mem. Left Compat.}) & \fa x \fa x' \fa y~
    (x=x' \land  x\in y\limp x'\in y) \\[6pt]
    (\text{Mem. Right Compat.}) & \fa x \fa y \fa y'~
    (y=y' \land  x\in y\limp x\in y') \\[6pt]
    (\text{Strong Extensionality}) &
    \fa x_1\cdots\fa x_n\fa a\fa b~ \\
    &\quad(R(a,b) \\
    &\quad\hphantom{(}{ \land }~\fa x\fa x'\fa y~
    (x'\in x \land  R(x,y)\limp\ex y'~(y'\in y \land  R(x',y'))) \\
    &\quad\hphantom{(}{ \land }~\fa y\fa y'\fa x~
    (y'\in y \land  R(x,y)\limp\ex x'~(x'\in x \land  R(x',y'))) \\
    &\quad\hphantom{(}{\limp}~a=b) \\
    \multicolumn{2}{l}
    {\qquad\text{for each formula $R(x,y)$ whose free variables
	are among $x_1,\ldots,x_n$, $x$ and $y$}} \\[6pt]
    (\text{Pairing}) & \fa a\fa b\ex e\fa x~
    (x\in e\liff x=a\lor x=b) \\[6pt]
    (\text{Union}) & \fa a\ex e\fa x ~
    (x\in e~\liff~\ex y~(x\in y \land  y\in a)) \\[6pt]
    (\text{Powerset}) & \fa a\ex e\fa x~
    (x\in e\liff x\subseteq a) \\[6pt]
    (\text{Restr. Comprehension}) &
    \fa x_1\cdots\fa x_n\fa a\ex e\fa x~
    (x\in e\liff x\in a \land  P(x)) \\
    \multicolumn{2}{l}
    {\qquad\text{for each formula $P(x)$ whose free variables
	are among $x_1,\ldots,x_n$, $a$ and $x$}} \\[6pt]
    (\text{Infinity}) & \ex e~\Ind(e) \\[6pt]
    (\text{Transitive closure}) & \fa a\ex e~(a\subseteq e \land 
    \fa x\fa y(x\in y \land  y\in e\limp x\in e)) \\[6pt]
    \hline\hline
  \end{array}$
  \caption{Axioms of the theory $\Zst$}
  \label{tbl:ZstAxioms}
\end{table}

Notice that in~$\Zst$ the standard formulation of the Extensionality
axiom is a consequence of the axiom of Strong Extensionality:
\begin{proposition} --- In $\Zst$, the following formula is provable:
  $$\fa a~\fa b~(\fa x~(x\in a\liff x\in b)~\limp~a=b)\,.
  \leqno(\text{Extensionality})$$
\end{proposition}

\begin{proof} 
Using the instance of strong extensionality where
the formula $R(x,y)$ is $(x=a \land y=b) \lor x=y$\,.\qed
\end{proof}

\subsection{A conservative extension with a sort for classes}

\begin{definition}[The theory $\Zclass$]
The theory $\Zclass$ is expressed in many-sorted predicate logic. 
It has two sorts $\Set$ and $\Class$.
Its language is formed with two binary predicate symbols $=$ and $\in$
of rank $\<\Set,\Set\>$ and a binary predicate symbol $\mem$ of rank 
$\<\Set,\Class\>$.
The axioms of the theory $\Zclass$ are 
\begin{itemize}
\item the axioms of equality of $\Zst$ and the axiom 
  $$\fa x~\fa y~\fa p~(x = y \land \mem(x,p) \Rightarrow \mem(y,p))$$
\item the strong extensionality scheme, generalized to all formul{\ae}
  possibly containing the symbol $\mem$ and free variables of sort
  $\Class$, but no quantification on classes;
\item the pairing axiom, the union axiom, the powerset axiom, the
axiom of infinity, the axiom of transitive closure;
\item the restricted comprehension scheme, generalized to all
  formul{\ae} possibly containing the symbol $\mem$ and free variables
  of sort $\Class$, but no quantification on classes;
\item and finally, a \emph{class comprehension scheme}
  $$\ex \alpha\ \fa x\ (\mem(x,\alpha) \liff P)$$
  for each formula~$P$ possibly containing the symbol $\mem$ and
  free variables of sort $\Class$, but no quantification on classes.
\end{itemize}
\end{definition}

All the axioms of $\Zst$ are axioms of $\Zclass$, thus $\Zclass$ is an
extension of $\Zst$. To prove that this is a conservative extension,
we use a notion of intuitionistic model where formul{\ae} are valuated
in a Heyting algebra~\cite{Rasiowa}, and we prove that for every
intuitionistic model of $\Zst$ there is an intuitionistic model of
$\Zclass$ validating the same formul{\ae} of the language of $\Zst$.
Conservativity follows from the correctness and completeness of
intuitionistic logic w.r.t.\ to its Heyting algebra valuated models.

\begin{definition}
  --- Let $\M$ be an intuitionistic model of $\Zst$, whose domain is
  still written~$\M$ and whose underlying Heyting algebra is
  written~$B$.
  A function~$E$ from~$\M$ to~$B$ is said to be \emph{definable} if
  there exists a formula~$P$ in the language of $\Zst$ whose free
  variables are among $x, y_{1},\ldots,y_{n}$ and elements
  $b_{1},\ldots,b_{n}$ of~$\M$ such that for all~$a$
  $\llbracket P \rrbracket_{a/x,b_1/y_1,\ldots,b_n/y_n} = E(a)$.
\end{definition}

\begin{definition}
  --- Let $\M$ be a model of $\Zst$, and consider the structure~$\M'$
  (with the same underlying Heyting algebra~$B$) defined as follows:
  $\Int{\Set}=\M$ and $\Int{\Class}$ is the set of definable
  functions from~$\M$ to the underlying algebra~$B$.
  The denotation of the symbols~$=$ and~$\in$ is the same as in~$\M$,
  and the denotation of the symbol~$\mem$ is function application.
\end{definition}

\begin{proposition}
  --- The structure~$\M'$ is a model of $\Zclass$.
\end{proposition}

\begin{proof}
  To prove that $\M'$ is a model of the class comprehension
  scheme, of the generalized extensionality scheme and of the the
  generalized restricted comprehension scheme, 
  we prove that for any formula~$P$ containing no quantifiers on
  variable of the sort $\Class$ and assignment $\phi$, there
  exists a formula~$Q$ in the language of $\Zst$ and an assignment
  $\phi'$ such that for all $a$,
  $$\llbracket P \rrbracket_{\phi + a/x} =
  \llbracket Q \rrbracket_{\phi' + a/x}$$
  We proceed by induction over the structure of $P$.
  The only non trivial case is when $P=\mem(x,p)$ where~$p$ and~$x$
  are variables.
  Then, the object $\llbracket p \rrbracket_{\phi}$ is a definable
  function from~$\M$ to~$B$.
  Let~$Q$ and~$\phi'$ be the defining formula and
  assignment, for all $a$, we have
  $$\llbracket P \rrbracket_{\phi + a/x} =
  \llbracket Q \rrbracket_{\phi' + a/x}\eqno\qed$$
\end{proof}

\medbreak
Obviously, a formula of the language of $\Zst$ has the same
denotation in $\M$ and in $\M'$, hence the conservativity of~$\Zclass$
over~$\Zst$.

\subsection{A conservative extension with Skolem symbols}
\label{Skolem}

The language of $\Zskol$ is the following.
Notice that the language of terms expressing sets and classes now
contains binding symbols.
$$\begin{array}{r@{~~}c@{~~}l}
  t,u &::=& x \quad|\quad \bigcup t \quad|\quad \{t_1,t_2\}
  \quad|\quad \P(t) \\
  &|& \{x\in t\mid P'\} \quad|\quad \N \quad|\quad \Cl(t) \\
  \noalign{\medskip}
  T,U &::=& X \quad|\quad \LBrc x\mid P'\RBrc \\
  \noalign{\medskip}
  P,Q &::=& t=u \quad|\quad t\in u \quad|\quad \mem(t,T) \\
  &|& \top \quad|\quad \bot \quad|\quad P \land  Q
  \quad|\quad P\lor Q \quad|\quad P\limp Q \\
  &|& \fa x~P \quad|\quad\ex x~P \quad|\quad
  \fa X~P \quad|\quad\ex X~P \\
  \noalign{\medskip}
  P',Q' &::=& t=u \quad|\quad t\in u \quad|\quad \mem(t,T) \\
  &|& \top \quad|\quad \bot \quad|\quad P' \land  Q'
  \quad|\quad P'\lor Q' \quad|\quad P'\limp Q' \\
  &|& \fa x~P' \quad|\quad\ex x~P' \\
\end{array}\leqno\begin{array}{@{}l}
  \textbf{Terms} \\ \\
  \noalign{\medskip}
  \textbf{Class terms} \\
  \noalign{\medskip}
  \textbf{Formul{\ae}} \\ \\ \\
  \noalign{\medskip}
  \textbf{Restricted} \\ \textbf{formul{\ae}} \\ \\
\end{array}$$

\begin{definition}
  We define three transformations:
  \begin{itemize}
  \item A transformation on \emph{terms}, which maps each term~$t$
    of~$\Zskol$ equipped with a variable~$z$ to a formula of $\Zclass$
    written $z\incirc t$;
  \item A transformation on \emph{class terms}, which maps each class
    term~$T$ of~$\Zskol$ equipped with a variable~$z$ to a formula of
    $\Zclass$ written $\memcirc(z,T)$;
  \item A transformation on \emph{formul{\ae}}, which maps each
    formula~$P$ of $\Zskol$ to a formula of $\Zclass$
    written~$P^{\circ}$.
  \end{itemize}
  These transformations are defined by the following equations:
  $$\begin{array}{l@{~~}c@{~~}l}
    z\incirc x &\equiv& z\in x \\
    z\incirc\bigcup t &\equiv&
    \exists y~(z\in y \land  y\incirc t) \\
    z\incirc\{t_1,t_2\} &\equiv& (z=t_1)^{\circ}\lor(z=t_2)^{\circ} \\
    z\incirc\P(t) &\equiv& \forall y~(y\in z\limp y\incirc t) \\
    z\incirc\{x\in t\mid P'\} &\equiv&
    z\incirc t \land {P'}^{\circ}(x \la z) \\
    z\incirc \N &\equiv& \Nat(z) \\
    z\incirc\Cl(t) &\equiv& \forall x~[
    \fa y_1\fa y_2~(y_1\in y_2 \land  y_2\in x\limp y_1\in x) \land {}\\
    && \hphantom{\forall x~[}
    \fa y~(y\incirc t\limp y\in x)~\limp~z\in x] \\
    \noalign{\medskip}
    \memcirc(z,X) &\equiv& \mem(z,X) \\
    \memcirc(z,\LBrc x \mid P'\RBrc) &\equiv&
    {P'}^{\circ}(x \la z) \\
    \noalign{\medskip}
    (t=u)^\circ &\equiv&
    \forall z~(z\incirc t\liff z\incirc u) \\
    (t\in u)^{\circ} &\equiv&
    \exists x~((x=t)^{\circ} \land  x\incirc u) \\
    (\mem(t,U))^{\circ} &\equiv&
    \exists x~((x=t)^{\circ} \land  \memcirc(x,U)) \\
    (P \land  Q)^{\circ} &\equiv& P^{\circ} \land  Q^{\circ} \\
    && \text{etc.} \\
    (\fa x~P)^{\circ} &\equiv& \fa x~P^{\circ} \\
    (\ex x~P)^{\circ} &\equiv& \ex x~P^{\circ} \\
  \end{array}$$
\end{definition}
Notice that if~$P$ is already in the language of $\Zclass$, then the
equivalence $P\liff P^{\circ}$ is (intuitionistically) provable in
$\Zclass$.  

The notion of provability in $\Zskol$ is defined by $\Zskol \vdash P$
if $\Zclass \vdash P^{\circ}$. An equivalent solution would be to
define provability in $\Zskol$ directly from the expected deduction
rules and from the skolemized versions of the axioms of $\Zclass$.

We shall use the following abbreviations:
$$\begin{array}{r@{~~}c@{~~}l}
  \varnothing &\equiv& \{x \in \N~|~\bot\}\\
  X \cup Y &\equiv& \bigcup \{X,Y\}\\
  \{a\} &\equiv& \{a,a\}\\
  \<a,b\> &\equiv& \{\{a\},\{a,b\}\}\\
  \pi_1(x) &\equiv& \bigcup\{x_1\in\bigcup x\mid
  \ex x_2\ \ x\equiv\<x_1,x_2\>\} \\
  \pi_2(x) &\equiv& \bigcup\{x_2\in\bigcup x\mid
  \ex x_1\ \ x\equiv\<x_1,x_2\>\} \\
  X \times Y &\equiv&
  \{z \in \P(\P(X \cup Y))~|~\ex x~\ex y~
  (x \in X \land y \in Y \land z = \<x,y\>\}\\
  0 &\equiv& \varnothing \\
  1 &\equiv& \{\varnothing\} \\
  f(x) &\equiv& \bigcup \{y \in \bigcup \bigcup f~|~\<x,y\> \in f\}\\
  f_{|D} &\equiv& \{c \in f~|~\pi_1(c) \in D\}\\
\end{array}$$

\subsection{Second-order class quantification}

The model construction of section~\ref{sec:Normalization} (which is
devoted to the normalization proof of $\Zermod$) is not done
relatively to the theory~$\Zskol$, but relatively to the
extension~$\ZskolS$ of $\Zskol$ in which we drop the restriction on
the formul{\ae} that may be used in set$/$class comprehension (thus
allowing class quantification to appear everywhere in the language).

Of course, $\ZskolS$ is definitely not a conservative extension
of~$\Zskol$.
Actually, it is a skolemized presentation of second-order
Zermelo set theory (extended with Strong Extensionality and
Transitive Closure), which is proof-theoretically stronger
than~$\Zst$.
However, $\ZskolS$ has an obvious extensional model in~$\ZF$
which is defined by setting
$$\llbracket\Set\rrbracket=V_{2\omega}\qquad\text{and}\qquad
\llbracket\Class\rrbracket=V_{2\omega+1}\,,$$
where $(V_{\alpha})$ denotes the cumulative hierarchy (indexed by
ordinals).

\subsection{Projective classes}

Let~$A$ be class defined by a formula $A(x)$ with at most one free
variable~$x$, and $\phi(x,y)$ a formula with at most two free
variables~$x$ and~$y$.
We say that $\phi$ is a projection onto~$A$ if the following
formul{\ae} are provable:
\begin{enumerate}
\item $\fa x~\ex y~\phi(x,y)$
\item $\fa x~\fa y~\fa y'~(\phi(x,y)\land\phi(x,y')\limp y=y')$
\item $\fa x~(A(x)\limp\phi(x,x))$
\item $\fa x~\fa y~(\phi(x,y)\limp A(y))$
\end{enumerate}
A class~$A$ is \emph{projective} if there is a projection~$\phi$
onto~$A$.
Notice that in classical set theory every nonempty class $A$ is
projective, by taking
$$\phi(x,y) ~\equiv~ (A(x)\land y=x)\lor(\lnot A(x)\land y=a)$$
where $a$ is an arbitrary object such that $A(a)$.
In intuitionistic set theory, it is not the case anymore.
In some case, the formula $\phi(x,y)$ can be written $y=t(x)$ for some
term~$t$ with at most one free variable~$x$.
In this case, conditions~1 and~2 vanish, and conditions 3 and~4 are
rephrased as:
\begin{enumerate}
\item[$3'$.] $\forall x~(A(x)\limp t(x)=x)$
\item[$4'$.] $\forall x~A(t(x))$
\end{enumerate}
In what follows, the term $t(x)$ will be written $\cast{x}_A^t$, or
simply $\cast{x}_A$ when the term~$t$ is clear in the context.

\section{A theory of pointed graphs}

\subsection{Informal presentation}

The main definition in this paper is the theory $\Zermod$ that is a
presentation of set theory in deduction modulo with rewrite rules
only, \emph{i.e.}\ with no axioms. At a first glance, the theory $\Zermod$
looks more like a theory of pointed graphs, where usual set theoretic
notions such as membership and equality are derived notions.

Informally, a pointed graph is just a pair formed with a directed graph and
a distinguished node, called the {\em root} of the pointed graph.
In the theory $\Zermod$, all pointed graphs share the same nodes, but
may have different roots and edges.
Thus we have a sort~$N$ for nodes and a sort~$G$
for pointed graphs. The main symbol of the theory is a ternary
predicate symbol $\eta$, the formula $x~\eta_{a}~y$ expressing that
there is a edge from $y$ to $x$ in the pointed graph $a$. We also have
a function symbol $\root$ mapping each pointed graph to its root.

The easiest way to represent a set as a pointed graph is to represent
it as a tree whose root is connected to the roots of the trees
representing the elements of the set.
For instance, the set $\varnothing$ is represented as a
pointed graph with no edges. The set $\{\varnothing\}$
is represented as a tree whose root has one child that has no children, etc.

In the figure below, the pointed graph with no
edges and root \texttt{1} is a representation of the set $\varnothing$,
the pointed graph with root \texttt{2} and the plain edge is a
representation of the set $\{\varnothing\}$ and the pointed graph with
root \texttt{4} and the dotted edges is a representation of the set
$\{\varnothing, \{\varnothing\}\}$.
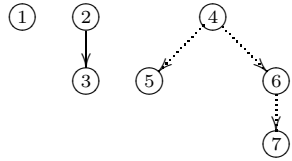
$$\def\circle#1{*+[o][F-]{\scriptstyle#1}}
\xymatrix@!=0pt{
  \circle{1} & \circle{2}\ar[d] &&
  \circle{4}\ar@{.>}[dl]\ar@{.>}[dr] &\\
  & \circle{3} & \circle{5} && \circle{6}\ar@{.>}[d] \\
  &&&& \circle{7} \\
}$$
We extend this idea by considering that any pointed graph represents a
set, namely, the set of objects represented by all the pointed graphs
obtained by shifting the root one level downwards.

Of course, a set may have several and non-isomorphic representations.
For instance, the graph with root \texttt{2} and plain edges and the
graph with root \texttt{6} and dotted edges both represent the set
$\{\varnothing\}$.
To recover the property of extensionality, we have to define equality
in such a way to identify these two pointed graphs.
Thus set equality is defined as bisimilarity.
Introducing a third sort for binary relations on nodes and a predicate
symbol $\rel$ (such that $\rel(x,y,r)$ means that $x$ and $y$ are
related by the relation~$r$), we then
define $a \approx b$ as
$$\begin{array}{l@{}l}
  \ex r~( & \rel(\root(a),\root(b),r) \\
  & \land \quad\fa x\fa x'\fa y~
  (x'~\eta_a~x \land \rel(x,y,r) \limp
  \ex y'~(y'~\eta_b~y \land \rel(x',y',r)))\\
  & \land \quad\fa y\fa y'\fa x~
  (y'~\eta_b~y \land \rel(x,y,r) \limp
  \ex x' ~ (x'~\eta_a~x \land \rel(x',y',r))))
\end{array}$$
In deduction modulo, this definition can be handled by introducing a
predicate symbol $\approx$ and a rule rewriting the atomic
formula $a \approx b$ to this one.

Next, we want to define the membership relation. We first introduce in the
language a binary function symbol $/$ associating a pointed graph to
each pair formed with a pointed graph and a node. The pointed graph
$a/x$ has the same graph as $a$ but its root is $x$.
This is expressed in deduction modulo by the rules
$$\begin{array}{c}
  \root(a/x) \lra x\qquad\quad (a/x)/y \lra a/y \\[3pt]
  x~\eta_{a/z}~y \lra x~\eta_{a}~y \\
\end{array}$$
Now, an object $a$ is a member of a set $b$ if the root of $b$ has a
child $x$ in $b$, such that $a$ is bisimilar to $b/x$.
In deduction modulo, this definition can be handled by introducing a
predicate symbol $\in$ and a rule
$$a \in b \lra \ex x~(x~\eta_{b}~\root(b) \land a \approx (b/x))$$
As equality on pointed graphs is not defined as the smallest
substitutive relation, but as bisimilarity, substitutivity has to be
proved.
In fact, equality is only substitutive with respect to the predicates
$\in$ and $\approx$, but not with respect to the symbol ``$/$'', for
instance.
Fortunately, substitutivity with respect to $\in$ and $\approx$ is
all we need to prove that $\Zermod$ extends $\Zst$. 

Equality on nodes is defined in a more usual way, introducing a
fourth sort for classes of nodes. The comprehension schemes
expressing the existence of classes and relations are handled  by
introducing a function symbol for each formula and the rewrite rules
$$\begin{array}{r>{\quad}c<{\quad}l}
  \mem(x,g_{x, y_{1}, ..., y_{n}, P}(y_{1}, \dots, y_{n}))
  &\lra& P \\
  \rel(x,x,',g'_{x, x', y_{1}, ..., y_{n}, P}(y_{1}, \dots, y_{n}))
  &\lra& P \\
\end{array}$$

Now, we want to build graphs for the usual set theoretic
constructions: pairing, union, powerset, restricted comprehension,
infinity and transitive closure. Let us take the example of the
union. If $a$ is a pointed graph, we want $\bigcup(a)$ to be a pointed
graph with a fresh root~$o$ related to all the grand children of the root
of~$a$. That the root of $\bigcup(a)$ is the node~$o$ can be expressed
in deduction modulo with the rule
$$\textstyle\root(\bigcup(a)) \lra o$$
Then, we want the formula
$x~\eta_{\bigcup(a)}~x'$ to hold if either $x$ and $x'$ are related in
the graph $a$ or $x'$ is $o$ and $x$ is a grand child of the root of
$a$. This could be expressed by the naive rule
$$\begin{array}{l}
  x~\eta_{\bigcup(a)}~x' \lra\\
  \qquad x~\eta_{a}~x' \lor \ex z~
  (x'=o \land  x~\eta_{a}~z \land z~\eta_{a}~\root(a))
\end{array}$$

However, with such as rule, we fail to express that
the root $o$ must be fresh. If it were already a node of $a$, for
instance, the properties of the set $\bigcup(a)$ would not be as
expected.
To build the pointed graph $\bigcup(a)$, we
must first {\em relocate} the graph $a$ in a space where there is no
$o$. This is achieved by introducing in the language a relocation
function $i$, that is injective but not surjective and a node $o$
that is not in the image of $i$.
Then the set $\bigcup(a)$ can be defined by the rule
$$\begin{array}{@{}l@{~~}l@{}}
  \multicolumn{2}{@{}l@{}}{x~\eta_{\bigcup(a)}~x'\ \lra} \\
  & (\ex y\ex y'~(x=i(y) \land  x'=i(y') \land  y~\eta_{a}~y')) \\
  \lor & (\ex y\ex z~(x=i(y) \land  x'=o \land  y~\eta_{a}~z
  \land z~\eta_{a}~\root(a)))\,. \\[3pt]
\end{array}$$
The fact that $i$ is injective is expressed in deduction modulo, following
\cite{DowekWerner2005} by introducing a left inverse $i'$ and
the rule
$$i'(i(x)) \ra x$$
To express that $o$ is not in the image of $i$,
we introduce a predicate $I$ that contains the image of $i$ but not
$o$. This is expressed by the rules
$$I(i(x)) \ra \top\qquad\quad I(o) \ra \bot$$
Some other constructions, such as pairing or powerset, need two relocation
functions $i$ and $j$ such that the images of $i$, $j$ and
$o$ are disjoint. To express the axiom of infinity, we also need a
copy of arithmetic at the level of nodes, thus we introduce
also symbols $Nat$, $0$,  $S$, $\Pred$, $\Null$, and $<$ and related
rules. For the powerset axiom, we need also an injection $\rho$
embedding pointed graphs into nodes.

\subsection{The theory $\Zermod$}

Let us now turn to the formal definition of~$\Zermod$.
The main symbol of this theory is a ternary predicate symbol $\eta$,
the formula $x~\eta_{a}~y$ meaning that there is a edge from~$x$
and~$y$ in the pointed graph~$a$.

The sorts of the theory~$\Zermod$ are the following:
\begin{center}
  \begin{tabular}{|>{$}p{12mm}<{$}|p{40mm}|}\hline
    \textbf{Sort} & \textbf{Usage} \\\hline
    G & pointed graphs \\
    N & Nodes \\
    C & Classes of nodes \\
    R & Binary relations on nodes \\\hline
  \end{tabular}
\end{center}
The predicate (\PRE), function (\FUN) and constant (\CST) symbols with
their arities are given in Table~\ref{tbl:SigZermod}.

The function symbols $f_{x,y_1,\ldots,y_n,P}$ are defined for each
formula $P$ with free variables $x,y_1,\ldots,y_n$ of sort $G$
formed with the predicate symbols $\in$ and $\approx$, and quantifiers
on $G$ only.
The function symbols $g_{x,y_1,\ldots,y_n,P}$ 
(resp. $g'_{x,x',y_1,\ldots,y_n,P}$) 
are defined for each
formula $P$ whose free variables are among 
$x, y_1, \ldots, y_n$ (resp. $x, x', y_1, \ldots, y_n$),
with $x$ (resp. $x,x'$) of sort $N$,
formed in the restriction of the language containing all the symbols,
except $g_{\cdots}$ and $g'_{\cdots}$.
The theory $\Zermod$ contains no axioms but the rewrite rules that
are given in Table~\ref{tbl:ZermodRewriteRules}.

\begin{table}
  \begin{tabular}{c}
    \hline\hline
    \declaresymbols{General}{
      \eta    & \PRE(G,N,N) & Local membership \\
      \root   & \FUN(G)N    & Root of a pointed graph \\
      /       & \FUN(G,N)G  & Change the root of a pointed graph \\
      =       & \PRE(N,N)   & Node equality \\
    } \\
    \declaresymbols{Sets and relations on nodes}{
      \mem    & \PRE(N,C)   & Node membership \\
      \rel    & \PRE(N,N,R) & Node relation \\
      g_{x,y_1,\ldots,y_n,P} & \FUN(N^n)C & Construction of sets of nodes \\
      g'_{x,x',y_1,\ldots,y_n,P} & \FUN(N^n)R &
      Construction of relations on nodes\\
    } \\
    \declaresymbols{Relocations}{
      o       & \CST\ N     & Distinguished node  \\
      i       & \FUN(N)N    & First injection \\
      i'      & \FUN(N)N    & Left-inverse of $i$ \\
      I       & \PRE(N)     & Image of $i$ \\
      j       & \FUN(N)N    & Second injection \\
      j'      & \FUN(N)N    & Left-inverse of $j$ \\
      J       & \PRE(N)     & Image of $j$ \\
      0       & \CST\ N     & zero  \\
      S       & \FUN(N)N    & successor\\
      \Pred   & \FUN(N)N    & Left-inverse of $S$ \\
      \Null   & \PRE(N)     & Singleton $0$ \\
      \Nat    & \PRE(N)     & Natural number nodes \\
      <  & \PRE(N,N)     & Strict ordering over nodes \\\hline
      \rho    & \FUN(G)N    & Injection from pointed graphs to nodes \\
      \rho'   & \FUN(N)G    & Left-inverse of $\rho$ \\
    } \\
    \declaresymbols{Equality and membership}{
      \approx & \PRE(G,G)   & Equality as bisimilarity \\
      \in     & \PRE(G,G)   & Membership as shifted bisimilarity \\
    } \\
    \declaresymbols{Constructions}{
      \bigcup   & \FUN(G)G   & Construction of the union \\
      \{\_,\_\} & \FUN(G,G)G & Construction of the pair \\
      \P        & \FUN(G)G   & Construction of the powerset \\
      f_{x,y_1,\ldots,y_n,P} & \FUN(G^n,G)G &
      Construction of sets by comprehension \\
      \Omega    & \CST\ G    & Pointed graph of Von Neumann numerals \\
      \Cl       & \FUN(G)G   & Construction of the transitive closure \\
    } \\
    \hline\hline
  \end{tabular}
  \caption{The signature of~$\Zermod$}
  \label{tbl:SigZermod}
\end{table}

\begin{table}[p]
  \begin{tabular}{c}
    \hline\hline\\[-4pt]
    \textbf{\underline{General}} \\
    $\begin{array}{r@{~}c@{~}l>{\qquad\quad}r@{~}c@{~}l}
      x~\eta_{a/z}~y &\lra& x~\eta_{a}~y & \root(a/x) &\lra& x \\
      y = z &\lra& \fa p~(\mem(y,p) \limp \mem(z,p)) &
      (a/x)/y &\lra& a/y \\
    \end{array}$ \\
    \noalign{\medskip}
    \textbf{\underline{Sets and relations on nodes}} \\[4pt]
    $\begin{array}{r>{\quad}c<{\quad}l}
      \mem(x,g_{x, y_{1}, ..., y_{n}, P}(y_{1}, \dots, y_{n}))
      &\lra& P \\
      \rel(x,x,',g'_{x, x', y_{1}, ..., y_{n}, P}(y_{1}, \dots, y_{n}))
      &\lra& P \\
    \end{array}$ \\
    \noalign{\medskip}
    \textbf{\underline{Relocations}} \\[4pt]
    $\begin{array}{r@{~}c@{~}l@{~~}r@{~}c@{~}l
        @{~}r@{~}c@{~}l@{~~}r@{~}c@{~}l}
      i'(i(x)) &\ra& x & I(i(x)) &\ra& \top &
      I(j(x)) &\ra& \bot & I(o) &\ra& \bot \\
      j'(j(x)) &\ra& x & J(j(x)) &\ra& \top &
      J(i(x)) &\ra& \bot & J(o) &\ra& \bot \\
     \Pred(S(x)) &\ra& x&
     \Null(0) &\ra& \top&
     \Null(S(x)) &\ra& \bot &
     \rho'(\rho(x)) &\ra& x \\
     \Nat(0) &\ra& \top &
     \Nat(S(x)) &\ra& \Nat(x) &
     x<0 &\ra& \bot & x<S(y) &\ra& x<y \lor x=y \\
    \end{array}$ \\
    \noalign{\medskip}
    \textbf{\underline{Equality and membership}} \\[6pt]
    $\begin{array}{r>{~~}c<{~~}l}
      a \approx b &\lra&
      \ex r~(\rel(\root(a),\root(b),r) \\
      &&\hphantom{\ex r~(} \land \quad\fa x\fa x'\fa y~
      (x'~\eta_a~x \land \rel(x,y,r) \limp
      \ex y' ~ (y'~\eta_b~y \land \rel(x',y',r))) \\
      &&\hphantom{\ex r~(} \land \quad\fa y\fa y'\fa x~
      (y'~\eta_b~y \land \rel(x,y,r) \limp
      \ex x' ~ (x'~\eta_a~x \land \rel(x',y',r)))) \\
      a \in b &\lra&
      \ex x~(x~\eta_{b}~\root(b) \land a \approx (b/x)) \\
    \end{array}$ \\

    \noalign{\medskip}
    \textbf{\underline{Constructions}} \\
    $\begin{array}{>{\qquad}ll}

      \multicolumn{2}{l}{x~\eta_{\bigcup(a)}~x'\ \lra} \\
      & (\ex y\ex y'~(x=i(y) \land  x'=i(y') \land  y~\eta_{a}~y')) \\
      \lor & (\ex y\ex z~(x=i(y) \land  x'=o \land  y~\eta_{a}~z
      \land z~\eta_{a}~\root(a))) \\[3pt]
      
      \multicolumn{2}{l}{x~\eta_{\{a,b\}}~x'\ \lra} \\
      &(\ex y \ex y'~(x = i(y) \land x' = i(y') \land y~\eta_{a}~y')) \\
      \lor& (\ex y \ex y'~(x = j(y) \land x' = j(y')
      \land y~\eta_{b} y')) \\
      \lor& (x = i(\root(a)) \land x' = o) \\
      \lor& (x = j(\root(b)) \land x' = o) \\[3pt]

      \multicolumn{2}{l}{x~\eta_{\P(a)}~x'\ \lra} \\
      &(\ex y \ex y'~(x = i(y) \land x' = i(y') \land y~\eta_{a}~y')) \\
      \lor& (\ex y \ex c~(x = i(y) \land x' = j(\rho(c))  \land 
      y~\eta_{a}~\root(a) \land (a/y) \in c)) \\
      \lor& (\ex c~(x = j(\rho(c)) \land x' = o)) \\[3pt]

      \multicolumn{2}{l}{x~\eta
        _{f_{x,y_{1},\ldots,y_{n},P}(y_{1},\dots, y_{n}, a)}~x' \lra} \\
      & (\ex y \ex y'~(x = i(y) \land x' = i(y')  \land 
      y~\eta_{a}~y')) \\
      \lor& (\ex y~(x = i(y) \land x' = o  \land 
      y~\eta_{a}~\root(a) \land P (x \la (a/y)))) \\[3pt]

      \multicolumn{2}{l}{x~\eta_{\Omega}~x'\ \lra} \\
      &(\ex y \ex y'~(x = i(y) \land x' = i(y') \land y < y')) \\
      \lor& (\ex y~(x = i(y) \land x' = o  \land \Nat(y))) \\[3pt]

      \multicolumn{2}{l}{x~\eta_{\Cl(a)}~x'\ \lra} \\
      & (\ex y\ex y'~(x=i(y) \land  x'=i(y') \land  y~\eta_{a}~y')) \\
      \lor & (\ex y~(x=i(y) \land x'=o \land \\
      & \hphantom{(\ex y~(}\forall c~
      [\forall z~(z~\eta_a~\root(a)\limp\mem(z,c)) \land {}\\
      & \hphantom{(\ex y~(\forall c~[}
       \fa z~\fa z'~((z~\eta_a~z' \land \mem(z',c)) \limp \mem(z,c))
        ~\limp~\mem(y,c)])) \\
    \end{array}$ \\

    \noalign{\smallskip}
    $\begin{array}{r>{~~}c<{~~}l>{\qquad\quad}r>{~~}c<{~~}l}
      \root(\bigcup(a)) &\lra& o &
      \root(\{a,b\}) &\lra& o \\
      \root(\P(a)) &\lra& o &
      \root(f_{x, y_{1}, ..., y_{n}, P}(y_{1}, \dots, y_{n}, a))
      &\lra& o \\
      \root(\Omega) &\lra& o &
      \root(\Cl(a)) &\lra& o \\
    \end{array}$ \\
    \noalign{\medskip}
    \hline\hline
  \end{tabular}
  \caption{Rewrite rules of $\Zermod$}
  \label{tbl:ZermodRewriteRules}
\end{table}

\begin{example}
Let $\varnothing = f_{x, y, \neg (x \in y)}(y,y)$.
\end{example}

\subsection{Translating $\Zst$ into $\Zermod$}
\label{sec:TranslZstZermod}

We prove that $\Zermod$ is an extension of set theory.  To do so, we
define a translation $P\mapsto P^{\dag}$ from $\Zst$ to $\Zermod$ 
which simply maps $\in$ (of $\Zst$) to $\in$ (of $\Zermod$) and
$=$ (of~$\Zst$) to $\approx$ (of~$\Zermod$), the rest of the
structure of the formula being preserved.  We then prove that
$\Zermod$ is an extension of~$\Zst$,
in the sense that for any formula $\phi$ of $\Zst$,
if $\Zst\vdash \phi$, then $\Zermod\vdash\phi^{\dag}$.

To prove this formula, we first prove that all axioms of $\Zst$
are theorems of $\Zermod$. We begin with fifty-three elementary
lemmas.

\begin{table}[p]
  $$\begin{array}{ll}
    \multicolumn{2}{c}{\textbf{\underline{Node identity}}} \\[6pt]
    1. & x = x \\
    2. & y = z \limp (P(x \la y) \limp P(x \la z))\quad(*) \\[12pt]
    \multicolumn{2}{c}{\textbf{\underline{Bisimilarity}}} \\[6pt]
    3. & a \approx a \\
    4. & a \approx b \limp b \approx a \\
    5. & (a \approx b \land b \approx c) \limp a \approx c \\
    6. & a \approx (a/\root(a))\\[12pt]
    \multicolumn{2}{c}{\textbf{\underline{Injectivity and
    non confusion}}} \\[6pt]
    7. & S(x) = S(y) \limp x = y \\
    8. & \neg 0 = S(x) \\
    9. & i(x) = i(y) \limp x = y \\
    10. & j(x) = j(y) \limp x = y \\
    11. & \neg i(x) = o \\
    12. & \neg j(x) = o \\
    13. & \neg i(x) = j(y) \\[12pt]
    \multicolumn{2}{c}{\textbf{\underline{Eta simplification}}} \\[6pt]
    14. & x~\eta_{\bigcup(a)}~i(y') \Leftrightarrow \ex y~(x = i(y)  \land 
    y~\eta_{a}~y')\\
    15. & x~\eta_{\bigcup(a)}~o \Leftrightarrow 
         \ex y~\ex z~(x=i(y) \land y~\eta_{a}~z \land z~\eta_{a}~\root(a))\\
    16. & x~\eta_{\{a,b\}}~i(y') \Leftrightarrow 
         \ex y~(x = i(y) \land y~\eta_{a}~y')\\
    17. & x~\eta_{\{a,b\}}~j(y') \Leftrightarrow 
         \ex y~(x = j(y) \land y~\eta_{b}~y')\\
    18. & x~\eta_{\{a,b\}}~o \Leftrightarrow 
         (x = i(\root(a)) \lor x = j(\root(b)))\\
    19. & x~\eta_{\P(a)}~i(y') \Leftrightarrow 
         \ex y~(x = i(y) \land y~\eta_{a}~y')\\
    20. & x~\eta_{\P(a)}~j(\rho(c)) \Leftrightarrow 
         \ex y ~(x = i(y) \land y~\eta_{a}~\root(a) \land (a/y) \in c)\\
    21. & x~\eta_{\P(a)}~o \Leftrightarrow \ex c~(x = j(\rho(c)))\\
    22. & x~\eta_{f_{x,y_{1},\ldots,y_{n},P}(y_{1},\dots, y_{n},
         a)}~i(y') \Leftrightarrow \ex y~(x = i(y)  \land 
         y~\eta_{a}~y')\\
    23. & x~\eta_{f_{x,y_{1},\ldots,y_{n},P}(y_{1},\dots, y_{n}, a)}~o
          \Leftrightarrow \ex y~(x = i(y) \land 
          y~\eta_{a}~\root(a) \land P (x \la (a/y)))\\
    24. & x~\eta_{\Omega}~i(y') \Leftrightarrow \ex y~(x = i(y)  \land 
         y < y') \\ 
    25. & x~\eta_{\Omega}~o \Leftrightarrow \ex y~(x = i(y)  \land \Nat(y))\\
    26. & x~\eta_{\Cl(a)}~i(y') \Leftrightarrow 
         \ex y~(x = i(y) \land y~\eta_{a}~y') \\
    27. &
      x~\eta_{\Cl(a)}~o \Leftrightarrow \\
      & \ex y~(x=i(y)  \land \\
      & \hphantom{(\ex y~(}\forall c~
      [\forall z~(z~\eta_a~\root(a)\limp\mem(z,c)) \land {}\\
      & \hphantom{(\ex y~(\forall c~[}
       \fa z~\fa z'~((z~\eta_a~z' \land \mem(z',c)) \limp \mem(z,c))
        ~\limp~\mem(y,c)]) \\[12pt]

    \multicolumn{2}{p{120mm}}{
      $(*)$ Where~$P$ is any formula of the language of $\Zermod$
      that contains no function symbol of the form~$g_{...}$
      or~$g'_{...}$.
    } \\

  \end{array}$$
\caption{}
\label{tbl:easy1}
\end{table}

\begin{table}[htp]
  $$\begin{array}{ll}
    \multicolumn{2}{c}{\textbf{\underline{Membership}}} \\[6pt]
    28. & x~\eta_{a}~\root(a) \limp (a/x) \in a \\
    29. & a \approx b \limp \fa x~(x~\eta_{a}~\root(a) \limp \ex
    y~(y~\eta_{b}~\root(b) \land (a/x) \approx (b/y))) \\
    30. & (a \in b \land a \approx c) \limp c \in b \\
    31. & (a \in b \land b \approx c) \limp a \in c \\[6pt]
    \multicolumn{2}{c}{\textbf{\underline{Substitutivity}}} \\[6pt]
    32. & (P(x \la a) \land a \approx b) \limp P(x \la b)\quad(*)
    \\[12pt]
    \multicolumn{2}{c}{\textbf{\underline{Bisimilarity by relocation}}} \\[6pt]
    33. & (\root(b) = i(\root(a)) \land 
    \fa x \fa y'~(y'~\eta_{b}~i(x) \liff 
    \ex x'~(y' = i(x') \land x'~\eta_{a}~x))) \limp a \approx b \\
    34. & (\root(b) = j(\root(a)) \land 
    \fa x \fa y'~(y'~\eta_{b}~j(x) \liff
    \ex x'~(y' = j(x') \land x'~\eta_{a}~x))) \limp a \approx b
    \\[12pt]
    \multicolumn{2}{c}{\textbf{\underline{Embedding}}} \\
    35. & \bigcup(a)/i(y) \approx (a/y) \\
    36. & (\{a,b\}/i(\root(a))) \approx a \\
    37. & (\{a,b\}/j(\root(b))) \approx b \\
    38. & \P(a)/i(y) \approx (a/y) \\
    39. & f_{x, y_{1}, ..., y_{p}, P}(a_{1}, ..., a_{p}, b)/i(y)
    \approx (b/y)\\
    40. & \Cl(a)/i(y) \approx (a/y)\\[12pt]
    \multicolumn{2}{c}{\textbf{\underline{Extensionality}}} \\
    41. & P(c,d) \\
    &  \land ~~(\fa a \fa a' \fa b~((a' \in a \land P(a,b)) \limp
    \ex b' (b' \in b \land P(a',b')))) \\
    &  \land ~~(\fa a \fa b \fa b'~((b' \in b \land P(a,b)) \limp
    \ex a' (a' \in a \land P(a',b')))) \\
    & \limp~~(c \approx d)\quad (*)\\[12pt]
    \multicolumn{2}{c}{\textbf{\underline{Finitary existence axioms}}}
    \\[6pt] 
    42. &  c \in \bigcup(a) \liff \ex b~(c \in b \land b \in a)\\
    43. &  c \in \{a,b\} \liff (c \approx a \lor c \approx b)\\
    44. &  a \in \P(b) \liff \fa c~(c \in a \limp c \in b)\\
    45. &  a \in f_{x, y_{1}, ..., y_{p}, P} (y_{1}, ..., y_{p},b)
    \liff a \in b \land P(x \la a)\quad(*)\\[12pt]
    \multicolumn{2}{c}{\textbf{\underline{Infinity}}} \\
    46. &  \neg a \in \varnothing\\
    47. &  \varnothing \approx (\Omega/i(0))\\
    48. &  (a \approx (\Omega/i(y))) \limp
    \bigcup(\{a,\{a\}\}) \approx (\Omega/i(S(y)))\\
    49. &  \varnothing \in \Omega\\
    50. &  a \in \Omega \limp \bigcup(\{a,\{a\}\}) \in \Omega\\
    51. &  \Ind(\Omega) \\[12pt]
    \multicolumn{2}{c}{\textbf{\underline{Transitive closure}}} \\
    52. &  a \in c \limp a \in \Cl(c)\\
    53. &  a \in b \limp b \in \Cl(c) \limp a \in \Cl(c)
    \\[12pt]
    \multicolumn{2}{p{115mm}}{
      $(*)$ Where $P$ is any formula expressed in the language
      $\approx$, $\in$ and where all the quantifiers are of sort~$G$.
    } \\
  \end{array}$$
\caption{}
\label{tbl:easy2}
\end{table}

\begin{theorem}
\label{extension}
If $\Zst \vdash P$ then $\Zermod \vdash P^{\dag}$. 
\end{theorem}

\begin{proof}
  We first prove the fifty three easy lemmas of
  tables~\ref{tbl:easy1} and~\ref{tbl:easy2}, from which we
  deduce that the axioms of $\Zst$ are provable in $\Zermod$.
  We conclude with a simple induction on proof structure.\qed
\end{proof}

\subsection{An example}

To understand the benefit of using pointed graphs and not directly
sets, take an arbitrary set~$A$ and consider the set~$C$ (built
using the restricted comprehension scheme) formed by the elements
of~$A$ that are not members of themselves.
The naive computation rule
$$a\in C ~~\lra~~ a\in A \land \lnot a\in a$$
makes the formula $C\in C$ reduce to $C\in A \land \lnot C \in C$.
Consequently, the rewrite system is non terminating, and the
underlying proof system is non normalizing too:
a simple adaptation of the proof of Russell's paradox yields a non
normalizable (but non paradoxical) proof of $\lnot C\in A$.

A simple attempt to solve the problem would be to replace the former
rule by
$$a\in C~~\lra~~\ex b~(b=a\land(b\in A\land\lnot b\in b))\,.$$
This way, the atomic formula $C\in C$ would reduce to the formula
$\ex b~(b=C\land(b\in A\land\lnot b\in b))$, and instead of the
formula $\lnot C\in C$ we would get the formula $\lnot b\in b$
(where~$b$ is a variable).
Using this trick, the rewrite system would be terminating, but we
could still build a non normalizable proof using the method
of~\cite{DowekWerner}.
The reason is that although we know that~$b$ is an element of~$A$ and
that~$A$ is structurally smaller than~$C$ (since~$C$ is built
from~$A$), nothing prevents us from substituting an arbitrary term to
the variable~$b$ in the sub-formula $\lnot b\in b$ during some
deduction step.

In $\Zermod$, in contrast, the formula $C\in C$ reduces to
$$\begin{array}{@{}l@{}l@{}l@{}}
  \ex x~(&(&
  \ex y\ex y'~(x=i(y)\land o=i(y')\land y~\eta_{A}~y')~~\lor\\
  && \ex y~(x = i(y) \land o = o \land  y~\eta_{A}~\root(A)
  \land \lnot (A/y) \in (A/y))) \\
  &\multicolumn{2}{@{}l@{}}{\land~~ C\approx (C/x))} \\
\end{array}$$
During this reduction step, we have evolved from $C\in C$ to
$\lnot(A/y)\in(A/y)$, where~$A$ is structurally smaller than~$C$.
This breaks the circularity and, in the normalization proof, we shall
be able to
interpret~$A$ first and then~$C$ according to this rule.

\section{Translating back $\Zermod$ into $\Zskol$}
\label{sec:TranslZermodZst}

To complete the proof that $\Zermod$ is actually a formulation of set
theory, we prove that it is a conservative extension of~$\Zst$.
Since $\Zskol$ is itself a conservative extension of $\Zst$, all we
need to prove is that $\Zermod$ is a conservative extension
of~$\Zskol$.
This proof is organized in two steps. First, we define a translation
$P\mapsto P^{*}$ from $\Zermod$ to $\Zskol$ and we prove that
if $\Zermod \vdash P$ then $\Zskol\vdash P^*$.
Then, we prove that the formula $P \Leftrightarrow P^{\dag*}$ is
provable in~$\Zskol$.

\subsection{Pointed graphs and reifications}

The translation $P\mapsto P^{*}$ is based on the fact that the
notions of pointed graph and bisimilarity can be defined in set
theory.

\begin{definition}[Pointed graph]
A (directed) \emph{graph} is a set of
pairs. A {\em pointed graph} is a pair $\<A,a\>$ where~$A$ is a graph.
\end{definition}

Notice that we do not include a carrier set in our graphs, since
the carrier~$\Car{A}$ of a graph~$A$ can always be reconstructed as
$$\textstyle\Car{A}=\bigl\{x\in\bigcup\bigcup A\mid\exists y~~
\<x,y\>\in A\lor \<y,x\>\in A\bigr\}\,,$$
whereas the carrier of a pointed graph
can be reconstructed as $\Car{\<A,a\>}=\Car{A}\cup\{a\}$.
The formula `$A$ is a graph' is
$$\Graph(A)~\equiv~\fa c\,{\in}\,A~\ex x~\ex y~c=\<x,y\>$$
and the formula `$g$ is a pointed graph' is
$$\Pgraph(g) \equiv \ex A\ \ex a\ (g = \<A,a\> \land \Graph(A))\,.$$

\begin{definition}[Bisimilarity]\label{def:bisim}
  --- Let $\<A,a\>$ and $\<B,b\>$
  be two pointed graphs.  A set $r$ is called a \emph{bisimulation}
  from $\<A,a\>$ to $\<B,b\>$ if
  \begin{enumerate}
  \item $\<a,b\>\in r$;
  \item for all~$x$, $x'$ and $y$ such that $\<x',x\>\in A$
    and $\<x,y\>\in r$, there exists~$y'$ such that
    $\<x',y'\>\in r$ and $\<y',y\>\in B$;
  \item for all~$y$, $y'$ and $x$ such that $\<y',y\>\in B$
    and $\<x,y\>\in r$, there exists~$x'$ such that
    $\<x',y'\>\in r$ and $\<x',x\>\in A$.
  \end{enumerate}
  Two pointed graphs $\<A,a\>$ and $\<B,b\>$ are said to be
  \emph{bisimilar} if there exists a bisimulation from 
  $\<A,a\>$ to $\<B,b\>$.
\end{definition}

Formally, the formula `$g$ and $g'$ are bisimilar' is
{\footnotesize
  $$\begin{array}{llll}
    g \approx g' & \equiv & \ex A \ex a \ex B \ex b \ex r~(\\
    &&\Graph(A) \land \Graph(B) \land g = \<A,a\> \land g' = \<B,b\> \\
    &&\land \<a,b\> \in r\\
    &&\land \fa x\fa x'\fa y~((\<x',x\> \in A \land \<x,y\> \in r) \limp
    \ex y'~(\<y',y\> \in B \land \<x',y'\> \in r)) \\ 
    &&\land~\fa y\fa y'\fa x~((\<y',y\> \in B \land \<x,y\> \in r) \limp
    (\ex x' ~ \<x',x\> \in A \land \<x',y'\> \in r)))
  \end{array}$$
}\ignorespaces
In the following definition, we will need a shorthand for
`$\phi$ is a function'
$$\begin{array}{r>{\quad}c<{\quad}l}
  \mathrm{Function}(\phi) &\equiv&
  \forall z\ (z\in\phi\ \limp
  \exists x\ \exists y\ \ z=\<x,y\>)\ \  \land {} \\
  &&\forall x\ \forall y\ \forall y'\
  (\<x,y\>\in\phi \land  \<x,y'\>\in\phi\limp y=y') \\
\end{array}$$
as well as terms $\Dom(\phi)$ and $\Cod(\phi)$ defined as
$$\begin{array}{r>{\quad}c<{\quad}l}
  \Dom(\phi) &\equiv& \{x\in\bigcup\bigcup\phi\mid
  \exists y\ \<x,y\>\in\phi\} \\
  \Cod(\phi) &\equiv& \{y\in\bigcup\bigcup\phi\mid
  \exists x\ \<x,y\>\in\phi\} \\
\end{array}$$

\begin{definition}[Collapse]
A \emph{Mostovski collapse} of a graph~$A$ is a function~$\phi$ of
domain $\Dom(\phi)=\Car{A}$ such that for any vertex $i\in\Dom(\phi)$
and for any $x$, we have $x \in \phi(i)$ if and only if there exists
$j\in\Dom(\phi)$ such that $\<j,i\>\in A$ and $x = \phi(j)$.
\end{definition}

Formally, the formula `$\phi$ is a collapse of~$A$' is
$$\begin{array}{l}
  \Collapse(A,\phi)\quad\equiv{} \\
  \qquad\Graph(A)\land \mathrm{Function}(\phi) \land
  \Dom(\phi)=\Car{A} \land {}\\
  \qquad\forall i\ \forall y'\ \fa y\
  [y' \in y \land \<i,y\> \in \phi\ \liff \exists i'\
  (\<i',i\> \in A\ \land \<i',y'\> \in \phi)]\\
\end{array}$$

The collapse of a graph, when it exists, is unique.  In ZF, this
property is a consequence of the Foundation axiom. However, the
weaker Strong Extensionality axiom is sufficient.

\begin{proposition}\label{prop:CollapseUnique} --- The
  formula
  $$\fa A~\fa\phi~\fa\psi~
  (\Collapse(A,\phi) \land \Collapse(A,\psi) \limp \phi=\psi)$$
  is derivable in $\Zskol$.
\end{proposition}

\begin{proof} 
  Let $A$ be a graph with two collapse functions $\phi$ and
  $\psi$. As a consequence of the instance of Strong Extensionality
  corresponding to the relation $r$ defined by
  $$r(u,v)\quad\equiv\quad \ex i~(\<i,u\> \in \phi \land
  \<i,v\> \in \psi)$$
  we get $x=x'$ for all $x$, $x'$ and $i$ such that
  $\<i,x\>\in\phi$ and $\<i,x'\>\in\psi$.\qed
\end{proof}

The domain of the collapse $\phi$ of a graph~$A$ is the carrier
$\Car{A}$ of~$A$. We extend it on the whole universe by introducing
the notation
$$\hat{\phi}_A(i) ~\equiv~
\{y\in\Cod(\phi)\mid\exists j~\<j,i\>\in A\land\<j,y\>\in\phi\}$$

\begin{proposition} --- The following formul{\ae} are provable in
  $\Zskol$:
  \begin{enumerate}
  \item $\fa A~\fa\phi~\fa i~
    (\Collapse(A,\phi)\land i\in\Car{A}~\limp~
    \phi(i) = \hat{\phi}_A(i))$
  \item $\fa A~\fa\phi~(\Collapse(A,\phi)~\limp~
    \fa i~\fa y~(y\in\hat{\phi}_A(i)~\liff~
    \ex j~(\<j,i\>\in A\land y=\hat{\phi}_A(j))))$
  \end{enumerate}
\end{proposition}

\begin{proof}
  \begin{enumerate}
  \item Assume that $\phi$ is a collapse of $A$ and 
    $i\in\Car{A}$. Then, by definition of $\hat{\phi}_A$, we have
    $\phi(i) = \hat{\phi}_A(i)$.
  \item If $\<j,i\> \in A$ then $j\in\Car{A}$, hence by the first
    part of the proposition, the formula
    $\<j,i\> \in A \wedge y = \hat{\phi}_A(j)$ is equivalent to 
    $\<j,i\> \in A \wedge y = \phi(j)$.\qed
  \end{enumerate}
\end{proof}

\begin{definition}[Reification]
  --- Let $\<A,a\>$ be a pointed graph whose underlying graph has a
  collapse~$\phi$. We say that an object $x$ is a reification
  of $\<A,a\>$ if $x=\hat{\phi}_A(a)$.

Formally, the formula `$x$ is a reification of $g$' is 
$$\Reif(g,x)\quad\equiv\quad \ex A~\ex a~\ex\phi~
(g=\<A,a\>\land\Collapse(A,\phi)\land x=\hat{\phi}_A(a))$$

The formula `$g$ is a reifiable pointed graph is'
$$\Rgraph(g) \equiv{} \ex x~\Reif(g,x)$$
As an immediate corollary of Prop.~\ref{prop:CollapseUnique} we get
the following proposition.
\end{definition}

\begin{proposition}\label{prop:ReifOfBisimPGraphs} The formula
  $$\fa g~\fa x~\fa y~((\Reif(g,x) \land \Reif(g,y)) \limp x = y)$$
  is derivable in $\Zst$.
\end{proposition}

\begin{proposition}\label{prop:PGraphsWithSameReif} --- The
  formula 
  $$\fa x~\fa g~\fa h~
  ((\Reif(g,x) \land \Reif(h,x)) \limp\ g \approx h)$$
  is derivable in $\Zst$.
\end{proposition}

\begin{proof}
Let $x$ be a set, and $g=\<A,a\>$ and $h=\<B,b\>$ be two pointed
graphs such that $\Reif(g,x)$ and $\Reif(h,x)$.  Assume that
$\phi$ is a collapse of $\<A,a\>$ such that $\phi(a)=x$ and
$\psi$ is a collapse of $\<B,b\>$ such that $\psi(b)=x$.
  We then define the relation $r$ by
  $$r=\{\<y,z\>\in\Dom(\phi)\times\Dom(\psi)\mid
  \hat{\phi}_{A}(y)=\hat{\psi}_{B}(z)\}$$
  and check that this is a bisimulation of~$g=\<A,a\>$ with
  $h=\<B,b\>$.\qed
\end{proof}

By definition, a reifiable pointed graph has a reification. 
We prove that, conversely, every set is the reification of some
pointed graph. This existence property can be proved with the
Replacement Scheme of $\ZF$. However, the weaker Transitive Closure
axiom is sufficient.

\begin{proposition}\label{prop:SetsAreReifs} The formula
$$\fa x\ \ex g~\Reif(g,x)$$
is derivable in $\Zst$.
\end{proposition}

\begin{proof} 
  Let $A$ be the set $\Cl(x) \cup \{x\}$ and $r$ the relation on $A$
  defined by  $r(u,v)$ if and only if $u \in v$, the set $x$ is the
  reification of the pointed graph $\<r,x\>$.\qed
\end{proof}

We now want to show that the class $\Rgraph$ is projective.
To do so, we first have to project any set~$A$ to a collapsible graph
$G(A)$.
Intuitively, the graph $G(A)$ is defined as the largest subgraph
of~$A$ that has a collapse.
This relies on the following definition:
\begin{definition}[Initial subgraph]
$$\ISeg(G,A) ~~\equiv ~~ \Graph(G)\land
  \fa x~\fa~y~((\<x,y\> \in A \land y \in\Car{G})
  \limp\<x,y\> \in G)$$
\end{definition}

\begin{proposition}\label{CollapseISeg}
If $\Collapse(A,\phi)$ and $\ISeg(G,A)$ then $\Collapse(G,
\phi_{|\Car{G}})$
\end{proposition}

\begin{proof}
Let $\psi = \phi_{|\Car{G}}$. It is routine to
check that if $i \in\Car{G}$, then the formul{\ae}
$$\ex j~(\<j,i\> \in A \land \<j,y\> \in \phi)
\qquad\text{and}\qquad
\ex j~(\<j,i\> \in G \land \<j,y\> \in \psi)$$
are equivalent. Thus, if $i \in\Car{G}$, then we have
$y \in \psi(i)$
if and only if 
$y \in \phi(i)$
if and only if 
$\ex j~(\<j,i\> \in A \land \<j,y\> \in \phi)$
if and only if 
$\ex j~(\<j,i\> \in G \land \<j,y\> \in \psi)$.
Thus $\psi$ is a collapse of $G$. \qed
\end{proof}

\begin{proposition}
Let $A$ be a graph, and $G_1$ and $G_2$ two initial subgraphs of $A$
with collapses $\phi_1$ and $\phi_2$. Then $\phi_1$ and $\phi_2$
coincide on $D = \Dom(\phi_1) \cap \Dom(\phi_2)$. 
\end{proposition}

\begin{proof}
  Let $G = A \cap (D \times D)$. 
  Notice that $D=\Car{G_1}\cap\Car{G_2}$. It is routine to check
  that $G$ is an initial subgraph of $G_1$ (resp. $G_2$).
  Let $D'=\Car{G}\subseteq D$.
  By Prop.~\ref{CollapseISeg}, ${\phi_1}_{|D'}$ and
  ${\phi_2}_{|D'}$ are collapses of $G$ hence
  they are equal by Prop.~\ref{prop:CollapseUnique}.
  We now want to prove that $\phi_1$ and $\phi_2$ coincide on the full
  set $D$. Consider an element $i\in D$.
  We have $y\in\phi_1(i)$ if and only if 
  $\exists j~\<j,i\>\in G_1\land y=\phi_1(j)$ if and only if 
  $\exists j~\<j,i\>\in G_2\land y=\phi_2(j)$ if and only if 
  $y\in\phi_2(i)$.
  The equivalence
  $\exists j~\<j,i\>\in G_1\land y=\phi_1(j)$ if and only if 
  $\exists j~\<j,i\>\in G_2\land y=\phi_2(j)$
  is justified by noticing that the proposition
  $\<j,i\>\in G_1$ and $\<j,i\>\in G_2$ are equivalent when $i\in D$
  (since both $G_1$ and $G_2$ are initial subgraphs of~$A$),
  and that in this case, we have $j\in D'$, hence
  $\phi_1(j)=\phi_2(j)$.\qed
\end{proof}

As an immediate corollary, we get:

\begin{proposition}
The union of all the initial subgraphs of a set $A$ that have a
collapse has a collapse.
\end{proposition}

\begin{definition}[Largest collapsible subgraph] --- The largest
  collapsible subgraph of a set~$A$ is given by
  $$G(A) = \bigcup \{G \in \P(A)~|~\ISeg(G,A) \land \ex
  \psi~\Collapse(G,\psi)\}$$
\end{definition}

The projection of any set~$x$ onto the class $\Rgraph$ of reifiable
pointed graphs is then defined as
$$\cast{x}_{\Rgraph} ~=~ \<G(\pi_1(x),\pi_2(x)\>$$

\subsection{Translation}

We are now ready to define a translation from $\Zermod$ to $\Zskol$.
Each sort $s$ of $\Zermod$ is interpreted as a sort of $\Zskol$
written $s_*$ accompanied with a relativization predicate written
$s^*(x)$ (where $x$ is of sort $s_*$).  We take
\begin{itemize}
\item $G_*=\Set$,\ \ with\ \ $G^*(x) \equiv \Rgraph(x)$
\item $N_*=\Set$,\ \ with\ \ $N^*(x) \equiv \top$
\item $C_*=\Class$,\ \ with\ \ $C^*(x) \equiv \top$
\item $R_*=\Class$,\ \ with\ \ $R^*(c) \equiv
  \fa x~(\mem(x,c)\limp\ex y~\ex z~(x=\<y,z\>))$
\end{itemize}

Each term $t$ (resp. formula~$P$) of $\Zermod$ is translated as a
term $t^*$ (resp. formula~$P^*$) of $\Zskol$.  These translations
are defined by mutual induction in Tables~\ref{tbl:AxInterp1}
and~\ref{tbl:AxInterp2}.

\begin{table*}
  $$\begin{array}{l}\hline\hline\\[-6pt]
    \begin{array}{rcl@{\qquad}rcl@{\qquad}rcl}
      \interp{x} &\equiv& x \\
      \interp{(\root(a))} &\equiv& \pi_2(a^*) &
      \interp{(i(a))} &\equiv& \<0,a^*\> &
      \interp{(j(a))} &\equiv& \<1,a^*\> \\
      \interp{(a/b)} &\equiv& \<\pi_1(a^*),b^*\> &
      \interp{(i'(a))} &\equiv& \pi_2(a^*) &
      \interp{(j'(a))} &\equiv& \pi_{2}(a^*) \\
      \interp{o} &\equiv& 0 &
      \interp{S(x)} &\equiv& \interp{x} \cup \{\interp{x}\} &
      \interp{(\rho(a))} &\equiv& a^* \\
      \interp{0} &\equiv& 0 &
      \interp{\Pred(x)} &\equiv& \bigcup \interp{x} &
      \interp{(\rho'(a))} &\equiv& \cast{a^*}_{\Rgraph} \\
    \end{array} \\
    \noalign{\bigskip}
    (g_{x,y_1,\ldots,y_n,P}(b_1,\ldots,b_n))^* \equiv
    \LBrc x\mid P^*(y_1 \la \interp{b_1}, ..., y_n\la\interp{b_n})\RBrc\\
    (g'_{x,x',y_1,\ldots,y_n,P}(b_1,\ldots,b_n))^* \equiv
    \LBrc z \mid \ex x \ex x'~(z = \<x,x'\> \land
    P^*(y_1 \la \interp{b_1}, ..., y_n \la \interp{b_n})) \RBrc\\
    \noalign{\bigskip}
    \interp{(\bigcup (a))} \equiv \<R,0\>~
    \mbox{where}~X \equiv (\{0\} \times \Car{a^*})\cup\{0\}\
    \mbox{and}\\
    \begin{array}{rcl}
      R \equiv \{c \in X \times X &|&
           \ex y \ex y'~(c=\<\<0,y'\>,\<0,y\>\> \land
           \<y',y\>\in \pi_1(a^*))\\
        && \lor \ex y' \ex y~(c=\<\<0,y'\>,0\> \land \<y',y\>\in
            \pi_1(a^*)\  \land  \<y,\pi_2(a^*)\>\in \pi_1(a^*))\} \\
    \end{array} \\[18pt]
    \interp{(\{a,b\})} \equiv \<R,0\>~
     \mbox{where}~X \equiv (\{0\}\times \Car{a^*})\cup(\{1\}\times
     \Car{b^*})\cup\{0\}~
     \mbox{and}\\
     \begin{array}{rcl}
      R \equiv \{c\in X \times X &|&
      \ex y \ex y'~(c=\<\<0,y'\>,\<0,y\>\> \land \<y',y\>\in \pi_1(a^*))\\
      && \lor \ex y \ex y'~(c=\<\<1,y'\>,\<1,y\>\> \land \<y',y\> \in
     \pi_1(b^*)) \\
      && \lor c=\<\<0,\pi_2(a^*)\>,0\> \lor c=\<\<1,\pi_2(b^*)\>,0\>\} \\
    \end{array} \\[18pt]
    (\P(a))^* \equiv \<R,0\>~
    \mbox{where}~ X \equiv (\{0\} \times
    \Car{a^*})\cup(\{1\}\times\P(\Car{a^*}))\cup\{0\}~\mbox{and}\\
     \begin{array}{rcl}
       R \equiv \{c\in X \times X &|&
       \ex y \ex y'\ (c=\<\<0,y'\>,\<0,y\>\> \land \<y',y\>\in \pi_1(a^*))\\
       &&\lor \ex y \ex p\ (c=\<\<0,y\>,\<1,p\>\> \land
       \<y,\pi_2(a^*)\>\in \pi_1(a^*) \land y \in p)\\
       &&\lor \ex p\ (c=\<\<1,p\>,0\>)\} \\
     \end{array}\\[18pt]
    (f_{x,y_1,\ldots,y_n,P}(a_1,\ldots,a_n,a))^* \equiv
     \<R,0\>~
     \mbox{where}~X \equiv (\{0\}\times \Car{a^*})\cup\{0\}~
     \mbox{and}\\
     \begin{array}{rcl}
       R \equiv \{c \in X \times X &|&
       \ex y \ex y'~
       (c=\<\<0,y'\>,\<0,y\>\> \land \<y',y\>\in\pi_1(a^*))\\
       && \lor~~ \ex y~(c=\<\<0,y\>,0\> \land
       \<y,\pi_2(a^*)\>\in\pi_1(a^*)\\
       &&\hphantom{\lor~~ \ex y~(}
       {\land}~P^*(x\la\<\pi_1(a^*),y\>,y_{1..n}\la a_{1..n}^*))\}\\
      \end{array}\\[18pt]
    \Omega^* \equiv \<R,0\>~
    \mbox{where}~
    X\equiv(\{0\}\times\N) \cup \{0\}~ \mbox{and}\\
    \begin{array}{rcl}
      R \equiv \{c \in X \times X &|&
      \ex y \ex y'~(c=\<\<0,y'\>,\<0,y\>\> \land y' \in y)\\
      && \lor \ex y\ (c = \<\<0,y\>,0\>)\} \\
    \end{array} \\[18pt]
    (\Cl(a))^* \equiv \<R,0\>~
    \mbox{where}~
    X\equiv(\{0\}\times\Car{a^*})\cup\{0\}~\mbox{and}\\
    \begin{array}{rcl}
      R \equiv \{c \in X \times X &|&
      \ex y \ex y'~(c=\<\<0,y'\>,\<0,y\>\> \land \<y',y\>\in\pi_1(a^*))\\
      && \lor \ex y~(c=\<\<0,y\>,0\> \land
      \<y, \pi_2(a^*) \> \in\Clos(\pi_1(a^*)))
    \end{array}\\[3pt]
    \text{where $\Clos(r)$ is the term}\\
    \{c\in\Car{r}\times\Car{r}\mid
    \fa r'~(r\subseteq r'~\land~
    \fa x\fa y\fa z~
    (\<x,y\>\in r'\land\<y,z\>\in r'\limp\<x,z\>\in r')
    ~\limp~c\in r')\} \\
    \noalign{\medskip}
    \hline\hline
  \end{array}$$
  \caption{Translation of terms}
  \label{tbl:AxInterp1}
\end{table*}

\begin{table}
  $$\begin{array}{r>{\quad}c<{\quad}l}
    \hline\hline
    (t~\eta_a~u)^* &\equiv& \<t^*,u^*\> \in \pi_1(a^*)\\
    (t = u)^* &\equiv& t^* = u^* \\
    (\mem(t,p))^* &\equiv& \mem(t^*,p^*)\\
    (\rel(t,u,r))^* &\equiv& \mem(\<t^*,u^* \>, r^*)\\
    (I(t))^* &\equiv&\exists y\ t^* =\<0,y\> \\
    (J(t))^* &\equiv&\exists y\ t^* =\<1,y\> \\
    (\Null(t))^* &\equiv&t^* = 0 \\
    (t < u)^* &\equiv&
    t^*\in u^* \land u^*\in\N \\
    (\Nat(t))^* &\equiv& t^*\in\N \\
    (t\approx u)^* &\equiv& t^* \approx u^*\\
    (t \in u)^* &\equiv& \exists z\ (\<z,\pi_2(u^*)\>\in\pi_1(u^*)\
    \land t^*\approx \<\pi_1(u^*), z\>)\\
    \top^* &\equiv& \top\\
    \bot^* &\equiv& \bot\\
    (A \limp B)^* &\equiv& A^* \limp B^*\\
    (A \land B)^* &\equiv& A^* \land B^*\\
    (A \lor B)^* &\equiv& A^* \lor B^*\\
    (\fa x~A)^* &\equiv& \fa x~(s_*(x) \limp A^*)\\
    (\ex x~A)^* &\equiv& \ex x~(s_*(x) \land A^*)\\
    \hline\hline
  \end{array}$$
  \caption{Translation of formul{\ae}}
  \label{tbl:AxInterp2}
\end{table}

\begin{proposition}\label{term} --- If $a$ is a well-formed term of sort~$s$
in~$\Zermod$ with free variables $x_1,\ldots,x_n$ of sorts
$s_1,\ldots,s_n$ respectively, then 
$$\Zskol \vdash \fa x_1\ \cdots\ \fa x_n\
    (s_1^*(x_1) \land  \cdots \land  s_n^*(x_n)
    \limp s^*(a^*))$$
\end{proposition}

\begin{proof}
By induction on the structure of the term $a$. The only non trivial
case is when $t$ is of sort $G$, in which case we have to check that 
$t^*$ is a term of sort $\Set$ and that the formula $\Reif(t)$ is
provable in $\Zskol$. If $a$ has the form $\bigcup(b)$, $\{b,c\}$, 
$\P(b)$, $f_{x,y_1,\ldots,y_n,P}(b_1,\ldots,b_n,b)$, $\Omega$ or
$\Cl(a)$, then we just apply the induction hypothesis and prove that
the pointed graph built in the translation is reifiable (which needs
to use the corresponding axioms of $\Zskol$). If $a$ has the form 
$b/x$ then we have to prove that the pointed graph built in the
translation is reifiable which is obvious because reifiability
does not depend on the position of the root in the graph.
If the term has the form $\rho'(a)$. 
We have to check that the 
pointed graph built in the
translation is reifiable, and this holds because $G(a)$ is built in
order to have a collapse.\qed
\end{proof}

\begin{proposition}[Correction of rules] 
\label{rules}
--- If $P \lra Q$, where the
  free variables of $P$ are among $x_1,\ldots,x_n$ of sorts
  $s_1,\ldots,s_n$ respectively, then the formula
  $$\Zskol \vdash s_1^*(x_1) \land  \cdots \land  s_n^*(x_n)\ \limp
  (P^*\liff Q^*)$$
\end{proposition}

\begin{proof}
We check this rule by rule. Let us give a few examples.
\begin{itemize}
\item The rule 
$$\begin{array}{lll}
x~\eta_{\bigcup(a)}~x'\ &\lra&
(\ex y~\ex y'~(x=i(y) \land  x'=i(y') \land  y~\eta_{a}~y')) \\
&& \lor~~(\ex y~\ex z~(x=i(y) \land  x'=o \land  y~\eta_{a}~z
      \land z~\eta_{a}~\root(a)))
\end{array}$$
Consider an atomic
formula of the form $t~\eta_{\bigcup(a)}~t'$ that reduces to 
$$\begin{array}{l}
(\ex y~\ex y'~(t = i(y) \land  t' = i(y') \land  y~\eta_{a}~y'))\\
\lor~~ (\ex y~\ex z~(t=i(y) \land  t'=o \land  y~\eta_{a}~z \land
z~\eta_{a}~\root(a)))
\end{array}$$
The translation of the formula 
$t~\eta_{\bigcup(a)}~t'$ is 
$\<\interp{t},\interp{t'}\> \in \pi_1(\<R,0\>)$
where
$$\begin{array}{l@{}l}
  R = \{ & c \in X \times X \mid\\
  & \hphantom{\lor}~~ \ex y~\ex y'~(c=\<\<0,y\>,\<0,y'\>\> \land
  \<y,y'\>\in \pi_1(a^*))\\ 
  & \lor~~ \ex y~\ex z~(c=\<\<0,y\>,0\> \land \<y,z\>\in
  \pi_1(a^*)\  \land  \<z,\pi_2(a^*)\>\in \pi_1(a^*))\}
\end{array}$$
where $X=(\{0\} \times \Car{a^*})\cup \{0\}$.
This formula is provably equivalent in $\Zskol$ to 
$$\begin{array}{l}
\interp{t} \in X \land \interp{t'} \in X \land\\
(\ex y~\ex y'~(\interp{t} = \<0,y\> \land \interp{t'} = \<0,y'\> \land
\<y,y'\>\in \pi_1(a^*))\\ 
\lor~~\ex y~\ex z~(\interp{t} = \<0,y\> \land \interp{t'}= 0 \land
\<y,z\>\in \pi_1(a^*)\land\<z,\pi_2(a^*)\>\in \pi_1(a^*)))
\end{array}$$
that is provably equivalent in $\Zskol$ to 
$$\begin{array}{l}
\ex y~\ex y'~(\interp{t} = \<0,y\> \land \interp{t'} = \<0,y'\> \land
\<y,y'\>\in \pi_1(a^*))\\ 
\lor~~ \ex y~\ex z~(\interp{t} = \<0,y\> \land \interp{t'}= 0 \land
\<y,z\>\in \pi_1(a^*)\land\<z,\pi_2(a^*)\>\in \pi_1(a^*))
\end{array}$$
that is the translation of
$$\begin{array}{l}
(\ex y~\ex y'~(t = i(y) \land  t' = i(y') \land  y~\eta_{a}~y'))\\
\lor~~(\ex y~\ex z~(t=i(y) \land  t'=o \land  y~\eta_{a}~z \land
z~\eta_{a}~\root(a)))\,.
\end{array}$$
\item The rule 
$$y = z \lra \fa p~(\mem(y,p) \limp \mem(z,p))$$
Consider an atomic formula $t = u$ that reduces to 
$$\fa p~(\mem(t,p) \limp \mem(t,p))\,.$$
The l.h.s.\ translates to the formula $t^*=u^*$
whereas the r.h.s.\ translates to
$$\fa p~(\top\limp\mem(t^*,p)\limp\mem(u^*,p))\,.$$
Both formul{\ae} are equivalent in $\Zclass$.
\couic{\item The rule
  $$\begin{array}{rcl}
    a \approx b &\lra&
    \ex r~(\rel(\root(a),\root(b),r) \\
    &&\hphantom{\ex r~(} \land \quad\fa x\fa x'\fa y~
    (x'~\eta_a~x \land \rel(x,y,r) \limp
    \ex y' ~ (y'~\eta_b~y \land \rel(x',y',r))) \\
    &&\hphantom{\ex r~(} \land \quad\fa y\fa y'\fa x~
    (y'~\eta_b~y \land \rel(x,y,r) \limp
    \ex x' ~ (x'~\eta_a~x \land \rel(x',y',r)))) \\
  \end{array}$$
  Consider an atomic formula $t\approx u$, which translates
  to $t^*\approx u^*$ (where $\approx$ now expresses bisimilarity in
  set theory in the sense of Def.~\ref{def:bisim}).
  Using the rule above, the formula $t=u$ reduces to a formula
  that translates to (TO FINISH)}
\end{itemize}
\end{proof}

\begin{proposition}[Correction of the translation] 
\label{toto}
--- Let $P$ be a
formula of $\Zermod$ with free variables $x_1,\ldots,x_n$ of
sorts $s_1,\ldots,s_n$ respectively.  If $\Zermod \vdash P$, then 
$$\Zskol \vdash s_1^*(x_1) \land  \cdots \land  s_n^*(x_n)\limp P^*$$
\end{proposition}

\begin{proof}
By induction over proof structure, using Prop.~\ref{term} to
justify the rules of quantifiers and 
Prop.~\ref{rules} to justify conversion steps.\qed
\end{proof}

\subsection{Conservative extension}
\label{subsec:ConsExt}

In Section \ref{sec:TranslZstZermod}, we have proved that 
$\Zermod$ was an extension of $\Zst$. We are now ready to prove that
this extension is conservative. 

\begin{proposition}\label{prop:ReifEquiv}
--- For any formula
  $P(x_1,\ldots,x_n)$ of~$\Zst$, the universal closure of the formula
  (with free variables $x_1,\ldots,x_n,g_1,\ldots,g_n$)
  $$\bigwedge_{i=1}^n
    \Reif(x_i, g_i)  \ \limp \
    \bigl(P(x_1,\ldots,x_n)\liff P^{\dag*}(g_1,\ldots,g_n)\bigr)$$
  is a theorem of $\Zskol$.
\end{proposition}

\begin{proof}
  By structural induction on~$P$.
  \begin{itemize}
  \item If $P(x,y)$ has the form $x = y$, let us assume 
    $\Reif(x,g)$ and $\Reif(y,h)$. We have to prove
    $$x=y \liff g\approx h$$
    and this is a consequence of Prop.~\ref{prop:PGraphsWithSameReif}
    and~\ref{prop:ReifOfBisimPGraphs}.
  \item If $P(x,y)$ has the form $x \in y$, let us assume 
    $\Reif(x,g)$ and $\Reif(y,h)$. 
    The formula $P^{\dag*}(g,h)$ is 
    $$\ex z~(\<z,\pi_2(h)\> \in \pi_1(h)  \land 
    g\approx \<\pi_1(h), z\>)$$ 
    and we have to prove the formula
    $$x \in y ~~\liff~~ \ex z~
    (\<z,\pi_2(h)\>\in\pi_1(h) \land g \approx \<\pi_1(h), z\>)$$
    Let $B=\pi_1(h)$, $b=\pi_2(h)$, and $a=\pi_2(g)$. We have to prove
    $$x \in y ~~\liff~~ 
    \ex z~(\<z,b\>\in B \land g \approx \<B, z\>)$$
    Let $\phi$ be a collapse of $g$ such that $x=\phi(a)$ and $\psi$ a
    collapse of~$h$ such that $y = \psi(b)$.
    \begin{itemize}
    \item Assume $x\in y$.  Then there exists $z\in\Dom(\psi)$ such that
      $\psi(z)=x$ and $\<z,b\>\in B$.  But $x$ is obviously a reification
      of the pointed graph $\<B,z\>$.  Since the pointed graphs~$g$ and
      $\<B,z\>$ have the same reification~$x$, they are bisimilar
      (Proposition~\ref{prop:PGraphsWithSameReif}).
    \item Conversely, assume $z$ such that $\<z,b\>\in B$ and
      $\<B,z\>\approx g$.  From $\<z,b\>\in B$, we get
      $\psi(z)\in\psi(b)=y$.  Since the pointed graphs~$\<B,z\>$ and~$g$
      are bisimilar, their reifications $\psi(z)$ and~$x$ are equal
      from~Proposition~\ref{prop:ReifOfBisimPGraphs}.
    \end{itemize}
  \item If $P(x_1,\ldots,x_n)$ has the form
    $Q(x_1,\ldots,x_n) \land  R(x_1,\ldots,x_n)$, then, by induction
    hypothesis, under the hypotheses 
    $\Reif(x_i,g_i)$, we have\\
    $Q(x_1,\ldots,x_n)\liff Q^{\dag*}(g_1,\ldots,g_n)$
    and $R(x_1,\ldots,x_n)\liff R^{\dag*}(g_1,\ldots,g_n)$. We
    deduce $(Q(x_1,\ldots,x_n) \land  R(x_1,\ldots,x_n))
    \liff (Q^{\dag*}(g_1,\ldots,g_n) \land 
    R^{\dag*}(g_1,\ldots,x_n))$, {\em i.e.}
    $(Q(x_1,\ldots,x_n) \land R(x_1,\ldots,x_n)) \liff
    (Q(g_1,\ldots,g_n) \land R(g_1,\ldots,g_n))^{\dag*}$.

  \item If $P(x_1,\ldots,x_n)$ has the form
    $Q(x_1,\ldots,x_n)\lor R(x_1,\ldots,x_n)$ or\\
    $Q(x_1,\ldots,x_n) \limp R(x_1,\ldots,x_n)$, the proof is
    similar.

  \item If $P(x_1,\ldots,x_n)$ has the form $\fa x~Q(x,x_1,\ldots,x_n)$,
    then $P^{\dag*}(g_1,\ldots,g_n)$ is
    $$\fa g~[\Rgraph(g)\limp Q^{\dag*}(g,g_1,\ldots,g_n)]\,.$$
    \begin{itemize}
    \item Let us assume $P(x_1,\ldots,x_n)$, \emph{i.e.}
      $\fa x~Q(x,x_1,\ldots,x_n)$, and prove\\
      $(P^{\dag})^*(g_1,\ldots,g_n)$, {\em i.e.}
      $\fa g~[\Rgraph(g)\limp Q^{\dag*}(g,g_1,\ldots,g_n)]$.\\
      Let $g$ be a reifiable pointed graph, and $a$ a reification
      of~$g$.  From our assumption, one has $Q(a,x_1,\ldots,x_n)$
      By induction hypothesis, we have $(Q^{\dag})^*(g,g_1,\ldots,g_n)$.
    \item Conversely, assume $(P^{\dag})^*(g_1,\ldots,g_n)$,
      \emph{i.e.}\\
      $\fa g~[\Rgraph(g)\limp Q^{\dag*}(g,g_1,\ldots,g_n)]$, and
      prove\\
      $P(x_1,\ldots,x_n)$, \emph{i.e.}
      $\fa x~Q(x,x_1,\ldots,x_n)$.  Let $x$ be a set.  From
      Prop.~\ref{prop:SetsAreReifs}, there exists a reifiable pointed
      graph~$h$ such that $\Reif(x,h)$.  By induction hypothesis we
      have $Q(x,x_1,\ldots,x_n)$.
    \end{itemize}
    
  \item If $P$ has the form $\ex x~Q$, the proof is similar.\qed
  \end{itemize}
\end{proof}

\begin{theorem}[Conservativity] --- 
Let $P$ be a closed formula in the language of $\Zst$. 
If $\Zermod \vdash P^{\dag}$, then 
$\Zst \vdash P$. 
\end{theorem}

\begin{proof}
Assume $\Zermod \vdash P^{\dag}$.
By Proposition~\ref{toto}, we have $\Zskol \vdash(P^{\dag})^*$,
by Proposition~\ref{prop:ReifEquiv} we get $\Zskol \vdash P$
and we conclude using the fact that $\Zskol$ is a conservative
extension of $\Zst$.\qed
\end{proof}

\section{Normalization}
\label{sec:Normalization}

In this section, we prove that all proofs in the theory~$\Zermod$ are
strongly normalizable.
As this theorem implies the consistency of~$\Zermod$, it cannot be
proved in set theory itself.
In~\cite{realizmod} we have generalized the usual notion of relative
consistency proof to a notion of relative normalization proof.
Technically, our normalization theorem is proved under the assumption
that $\ZskolS$ is $1$-consistent.

\subsection{Reducibility candidates}

To prove normalization, we shall use the result proved in
\cite{realizmod}.
For that, we need to define a translation from $\Zermod$ to $\ZskolS$
associating to each term $t$ of $\Zermod$ a term $t^*$ of $\ZskolS$
and to each atomic formula $P$ of $\Zermod$ a formula  $\pi \real P$
of $\ZskolS$.
This translation is then extended to all formul{\ae} as shown
in~\cite{realizmod}.
To define the formula $\pi \real P$ we shall first define a term
$P^*$ expressing a reducibility candidate and then we shall define
$\pi \real P$ as $\pi \in P^*$. 

We refer to \cite{realizmod} for the definition of all notations
related to reducibility candidates.
In particular, we shall denote
$\Proof$ the set of all proof-terms,
$\CR$ the set of all reducibility candidates,
$\SN$ the set of all strongly normalizable proofs
(which is the largest reducibility candidate),
and  $\tilde{\limp}$, $\tilde{\land}$, $\tilde{\lor}$, etc.
the binary operations on $\CR$ that interpret the corresponding
intuitionistic connectives.

An important property of the class of reducibility candidates is
that it is projective.
Indeed, if we define $\cast{X}_{\CR}$ as the intersection of
all reducibility candidates containing $X\cap\SN$
$$\cast{X}_{\CR} ~=~ \{\pi\in\SN\mid\forall r{\in}\CR~
(X\cap \SN\subseteq r\limp\pi\in r)\}$$
we easily check that $\ZskolS$ proves
\begin{enumerate}
\item For all $X$, $\cast{X}_{\CR}\in\CR$
\item If $X\in\CR$, then $X=\cast{X}_{\CR}$.
\end{enumerate}
Moreover, if $X$ is a set of strongly normalizable proofs,
then $\cast{X}_{\CR}$ is the smallest reducibility candidate
containing~$X$.

\subsection{Saturated pointed graphs}

\begin{definition}[Saturated pointed graph]
  A \emph{saturated graph} is a function $R$ whose
  domain is a set of pairs and whose codomain is $\CR$.
  A \emph{saturated pointed graph} is a pair $\<R,r\>$
  formed by a saturated graph~$R$ and an arbitrary object~$r$.
\end{definition}

The formul{\ae} `$x$ is a saturated graph' and
`$x$ is a saturated pointed graph' are written
$\Sgraph(x)$ and $\Spgraph(x)$, respectively.
Again, it is easy to check that
the class of saturated graphs and
the class of saturated pointed graphs
are projective, using the projections:
$$\begin{array}{rcl}
  \cast{X}_{\Sgraph} &\equiv&
  \{c\in\Dom(X)\times\CR\mid
  \pi_2(c)=\cast{X(\pi_1(c))}_{\CR}\} \\
  \cast{X}_{\Spgraph} &\equiv&
  \<\cast{\pi_1(X)}_{\Sgraph},~\pi_2(X)\> \\
\end{array}$$
We check that $\ZskolS$ proves
\begin{enumerate}
\item For all $X$, $\Spgraph(\cast{X}_{\Spgraph})$
\item If $\Spgraph(X)$, then $X=\cast{X}_{\Spgraph}$.
\end{enumerate}
(and similarly for $\Sgraph$).

The carrier $\Car{a}$ of a saturated pointed graph $a$ is defined as
$$\textstyle\Car{a} ~~\equiv~~
\{x\in\bigcup\bigcup\pi_1(x)~|~\ex y~\ex r~
\<\<x,y\>,r\>\in\pi_1(a)\lor\<\<y,x\>,r\>\in\pi_1(a)\}\,.$$

\subsection{Translation of sorts}

We now define the translation of $\Zermod$ into $\ZskolS$.

Each sort $s$ of $\Zermod$ is translated as a sort of $\ZskolS$
written $s_*$ accompanied with a relativization predicate written
$s^*(x)$ (where $x$ is of sort $s_*$).  We then set:
\begin{itemize}
\item $G_*=\Set$,\ \ with\ \ $G^*(x)\equiv \Spgraph(x)$
\item $N_*=\Set$,\ \ with\ \ $N^*(x)\equiv\top$
\item $C_*=\Class$,\ \ with
  $$C^*(c) ~~\equiv~~
  \fa z~(\mem(z,c) \limp \ex x~\ex r~(z=\<x,r\>\land r\in\CR))$$
\item $R_*=\Class$,\ \ with
  $$R^*(c) ~~\equiv~~
  \fa z~(\mem(z,c) \limp \ex x~\ex y~\ex r~(z=\<\<x,y\>,r\>
  \land r\in\CR))$$
\end{itemize}

If $c$ is an element of $C^*$ and if $x$ is any object, we write
$$\textstyle c[x]~~=~~\cast{\bigcup\{r\in\CR\mid
  \mem(\<x,r\>,c)\}}_{\CR}$$
the candidate associated to $x$ in $c$ (or the smallest candidate if
there is no candidate associated to $x$ in $c$).
Similarly, if $c$ is an element of $R^*$ and if $x,y$ are arbitrary
objects, we write
$$\textstyle c[x,y]~~=~~\cast{\bigcup\{r\in\CR\mid
  \mem(\<\<x,y\>,r\>,c)\}}_{\CR}$$
the candidate associated to $\<x,y\>$ in $c$ (or the smallest
candidate otherwise).

\subsection{Translation of function and predicate symbols}

To each function symbol $f$ of arity $n$ of $\Zermod$, we associate 
a term $\tilde{f}(x_1, ..., x_n)$ possibly containing the free variables
$x_1, ..., x_n$.
These ``macros'' will be used later to translate full terms, setting
$(f(t_1,..., t_n))^* \equiv \tilde{f}(t_1^*, ..., t_n^*)$.

We start by some easy function symbols:
$$\begin{array}{r@{~~}c@{~~}l@{\qquad}r@{~~}c@{~~}l@{\qquad}r@{~~}c@{~~}l}
  \tilde{\root}(x) &\equiv& \pi_2(x) &
  x\tilde{/}y &\equiv& \<\pi_1(x),y\> &
  \qquad \tilde{o}~&\equiv&~0 \\[6pt]
  \tilde{i}(x) &\equiv& \<0,x\> &
  \tilde{j}(x) &\equiv& \<1,x\> &
  \tilde{\rho}(x) &\equiv& x\\[6pt]
  \tilde{i'}(x) &\equiv& \pi_2(x) &
  \tilde{j'}(x) &\equiv& \pi_2(x) &
  \tilde{\rho'}(x) &\equiv& \cast{x}_{\Spgraph} \\[6pt]
  \qquad\tilde{0}~&\equiv&~0 &
  \tilde{S}(x) &\equiv& x \cup\{x\} &
  \tilde{\Pred}(x) &\equiv& \bigcup x \\
\end{array}$$

In the same way, to each predicate symbol $P$ of arity $n$ of
$\Zermod$, we associate a term $\tilde{P}(x_1, ..., x_n)$ 
possibly containing the free variables $x_1, ..., x_n$.
From these macros we will translate atomic formul{\ae}
by setting
$(P(t_1,\ldots,t_n))^*\equiv\tilde{P}(t_1^*,\ldots,t_n^*)$.
We set
$$\tilde{\mem}(x,p) ~~\equiv~~ p(x)$$
$$\tilde{\rel}(x,y,p) ~~\equiv~~ p(x,y)$$
$$\tilde{I}(x) ~~\equiv~~ \tilde{J}(x) ~~\equiv~~ \tilde{\Null}(x)
~~\equiv~~ \tilde{\Nat}(x) ~~\equiv~~ \SN$$  
$$(x~\tilde{\eta}_{a}~y) ~~\equiv~~ \cast{\pi_{1}(a)(x,y)}_{\CR}$$ 

Using both definitions 
$(f(t_1,..., t_n))^* \equiv \tilde{f}(t_1^*, ..., t_n^*)$ and 
$(P(t_1,..., t_n))^* \equiv \tilde{P}(t_1^*, ..., t_n^*)$, 
we can now translate all the terms and formul{\ae} containing the
symbols for which we have already introduced a translation.
In particular, we can translate the formula
$\fa p~(\mem(x,p) \Rightarrow \mem(y,p))$. 
We thus translate the predicate symbol $=$ as 
$$x~\tilde{=}~y \equiv \interp{(\fa p~(\mem(x,p) \Rightarrow \mem(y,p)))}$$
In a similar way we take
$$\begin{array}{@{}r@{~~}c@{~~}l@{}l@{}}
  a~\tilde{\approx}~b &\equiv&
  {[}\ex r~( & \rel(\root(a),\root(b),r) \\
    &&& \begin{array}{@{}l@{~}l@{}l@{}}
      \land & \fa x\fa x'\fa y\,( &
      x'\,\eta_a\,x \land \rel(x,y,r) \limp \\
      && \ex y'\,(y'\,\eta_b\,y \land \rel(x',y',r))) \\
      \land & \fa y\fa y'\fa x\,( &
      y'\,\eta_b\,y \land \rel(x,y,r) \limp \\
      && \ex x'\,(x'\,\eta_a\,x \land \rel(x',y',r)))){]}^* \\
    \end{array} \\
  a~\tilde{\in}~b &\equiv&
  {[\ex x~(} & {x~\eta_{b}~\root(b) \land a \approx (b/x))]^*} \\
\end{array}$$

To define $x~\tilde{<}~y$ and $(t < u)^*$, we proceed as follows.
Fix $x_0$ and $y_0$, and consider the sequence of functions
$(f_n)_{n\in\N}:\Cl(\{y_0\})\to\CR$ defined by induction on $n$ as
follows:
\begin{itemize}
\item $f_0$ is defined as the constant function that maps all the
  elements of $\Cl(\{y_0\})$ to the smallest candidate.
\item $f_{n+1}$ is defined from $f_n$ in two steps as follows:
  \begin{itemize}
  \item First we consider the functional graph $f'_n$ defined as
    $$f'_n=\{(0,\bot)\}\cup
    \{\<s(z),~f_n(z)\tlor x_0=z\>\mid y\in\Cl(\{y\})\}$$
  \item Then we set 
    $$f_{n+1}(z)=\cast{\{\pi\in\Proof\mid
    \exists c~(\<z,c\>\in f'_n \land \pi\in c)\}}_{\CR}$$
    for all $z\in\Cl(\{y\})$.
  \end{itemize}
\end{itemize}
Finally we set\quad
$x_0~\tilde{<}~y_0~~\equiv~~
\cast{\bigcup_{n\in\N}f_n(y_0)}_{\CR}$.

We then translate the symbols $\bigcup$, $\{,\}$, $\P$, 
$f_{x,y_1, ..., y_n}$, $\Omega$ and $\Cl$.

The formula $x~\eta_{\bigcup(a)}~x'$ reduces to the formula $P$
which is:
$$\begin{array}{l}
  (\ex y~\ex y'~(x=i(y) \land  x'=i(y') \land y~\eta_{a}~y')) \\
  \lor~~ (\ex y~\ex z~(x=i(y) \land  x'=o \land  y~\eta_{a}~z
  \land z~\eta_{a}~\root(a)))
\end{array}$$
Consider the translation $P^*$ of this formula.
We let 
$$\textstyle\tilde{\bigcup}(a) \equiv \<R,0\>$$
where
$R=\{c\in(X\times X)\times\CR\mid\ex x~\ex x'~
c=\<\<x,x'\>,P^*\>\}$
and $X=(1\times \Car{a})\cup\{0\}$.
We do the same thing for the other constructions.

Finally, remains to define the translation of the symbols 
$g_{x,y_1,\ldots,y_n,P}$ and $g'_{x,x',y_1,\ldots,y_n,P}$.
We set
$$\tilde{g}_{x,y_1,\ldots,y_n,P}(y_1,\ldots,y_n)~~\equiv~~
\LBrc z ~|~ \ex x~z= \<x, P^*\>\RBrc$$
$$(\tilde{g}'_{x,x',y_1,\ldots,y_n,P}(y_1,\ldots,y_n))^*~~\equiv~~
\LBrc z ~|~ \ex x~\ex x'~z= \<\<x,x'\>, P^*\>\RBrc$$

From \cite{realizmod}, to get normalization, we need to prove the
following two lemmas:

\begin{proposition} --- For any atomic formula $A$ of $\Zermod$,
the formula $A^* \in \CR$ is provable in $\ZskolS$.
\end{proposition}

\begin{proof} By induction on the structure of~$A$.\qed
\end{proof}

\begin{proposition} --- If $A \lra B$, then 
$A^* = B^*$ is provable in $\ZskolS$ under the 
assumptions $s_i^{*}(x_i)$ for each variable $x_i$ of sort $s_i$ that 
appears in  one of the formul{\ae}~$A$ and~$B$.
\end{proposition}

\begin{proof}  It suffices to prove the formula for each rewrite
  rule $A\lra B$ (for which~$A$ is always an atomic formula).  In most
  cases, this is obvious, since the denotation of the left-hand side
  has been precisely defined as the denotation of the right-hand
  side.\qed
\end{proof}

Thus we get our final theorem.

\begin{theorem} --- 
  If $\ZskolS$ is $1$-consistent, then the theory $\Zermod$
  has the normalization property.
\end{theorem}

\section{Witness properties}

\begin{corollary}[Witness property in $\Zermod$]
  \label{c:WitnessZermod}
  If a closed formula $\exists x~P(x)$ is provable in $\Zermod$,
  then there exists a term~$t$ in~$\Zermod$ (of the same sort as the
  variable~$x$) such that the formula $P(t)$ is provable.
\end{corollary}

\begin{proof} A cut-free proof ends with an introduction rule.
\end{proof}

\begin{corollary}[Non-numerical witness in~$\Zst$]
  If a closed formula $\exists x~P(x)$ is provable in~$\Zst$,
  then there exists a formula $D(y)$ with one free variable~$y$
  such that
  \begin{enumerate}
  \item The formula $\exists!x~D(x)$ is provable in $\Zst$.
  \item The formula $\forall x~(D(x)\limp P(x))$
    is provable in $\Zst$.
  \end{enumerate}
\end{corollary}

\begin{proof}
  The formula $\exists x~P^{\dagger}(x)$ is provable in~$\Zermod$,
  hence by corollary~\ref{c:WitnessZermod} there exists a term~$t$
  such that $P^\dagger(t)$ is provable in~$\Zermod$.
  Hence the formula $P^{\dagger*}(t^*)$ is provable in~$\Zst$.
  Consider the formula $D(y)\equiv\Reif(y,t^*)$.
  From Prop.~\ref{term}, we have $\Rgraph(t^*)$, that
  is: $\exists x~D(x)$.
  Uniqueness follows from Prop.~\ref{prop:ReifOfBisimPGraphs}.
  From Prop.~\ref{prop:ReifEquiv}, we have
  $$\forall x~\forall g~(\Reif(x,g)\limp
  (P(x)\liff P^{\dagger*}(g)))$$
  hence
  $$\forall x~(\Reif(x,t^*)\limp(P^{\dagger*}(t^*)\limp P(x)))$$
  As we have $P^{\dagger*}(t^*)$, we get
  $\forall x~(D(x)\limp P(x))$.\qed
\end{proof}

\begin{corollary}[Numerical witness in~$\Zst$]
  If a closed formula of the form
  $\exists~x(\Nat(x)\land P(x))$ is provable
  in~$\Zst$, then there exists a natural number $n$ such that
  the formula
  $$\exists x~(\mathrm{Is}_n(x)\land P(x))$$
  is provable in $\Zst$, 
  where the formula $\mathrm{Is}_n(x)$ is inductively defined by
  $$\mathrm{Is}_0(x)\equiv\Empty(x)\qquad\text{and}\qquad
  \mathrm{Is}_{n+1}(x)\equiv
  \exists y~(\mathrm{Is}_n(y)\land\Succ(y,x))$$ 
\end{corollary}

\begin{proof}
  The formula 
  $\exists~x(\Nat^{\dagger}(x)\land P^{\dagger}(x))$ is provable 
  in $\Zermod$, hence there exists a term $t$ in $\Zermod$ such that 
  $\Nat^{\dagger}(t)$ and $P^{\dagger}(t)$ are provable. 
  We check that the formula 
  $\fa x~(Nat^{\dagger}(x) \limp x \in \Omega)$ is provable in
  $\Zermod$. Hence the formula $t\in\Omega$, i.e.
  $\ex x~(x~\eta_{\Omega}~o)\land t\approx(\Omega/x)$ is provable.
  Again there exists a term~$u$ such that the formulae
  $u~\eta_{\Omega}~o$ and $t \approx (\Omega / u)$ are provable.
  The formula $u~\eta_{\Omega}~o$ is equivalent by elementary means
  to $\ex y~(u = i(y) \land Nat(y))$.
  
  Thus there exists a term $v$ such that $u=i(v)$ and $\Nat(v)$
  are provable.
  
  A cut free proof of the formula $\Nat(v)$ ends with an introduction
  rule. Hence $\Nat(v)$ reduce to a non atomic formula and~$v$ has the
  form $S^n(0)$ for some~$n$.

  Thus the formula $t\approx(\Omega/i(S^n(0)))$ is provable.

  We check, by induction on $n$ that the formula
  $\mathrm{Is}_n^{\dagger}(\Omega/i(S^n(0)))$ is provable
  in~$\Zermod$.
  Hence the formula
  $\ex x~(\mathrm{Is}_n^{\dagger}(x) \land P^{\dagger}(x))$ is
  provable in~$\Zermod$. Hence 
  $\ex x~(\Rgraph(x)\land\mathrm{Is}_n^{\dagger*}(x) \land
  P^{\dagger*}(x))$ is provable in $\Zst$ and by
  Prop. \ref{prop:ReifEquiv}, 
  the formula 
  $\ex x~(\mathrm{Is}_n(x) \land P(x))$ is provable
  in $\Zst$.
\end{proof}

\section{Conclusion}

In this paper we have given a normalization proof for Zermelo set
theory extended with Strong Extensionality and Transitive Closure.

This theorem can also be attained as a corollary of the existence of a
translation of $\Zst$ into type theory~\cite{Miq04}
(using stronger assumptions than the $1$-consistency of~$\ZskolS$).
However, instead of expressing set theory on top of a theory of graphs
defined in type theory, we expressed it on top of a theory of
graphs simply expressed in predicate logic.
The fact that this theory can be expressed with computation rules only
and no axioms is a key element for the cut-free proof to end with an
introduction rule.
This shows that the key feature of type theory used in these
translations is the feature captured by Deduction modulo: the
possibility to mix computation and deduction.

Along the way we have proposed a typed lambda-calculus where all
terms normalize and where all provably total functions of set theory
can be expressed. It should be noticed that the syntax of
lambda-calculus is exactly that of the proofs of predicate logic
(\emph{i.e.} variables, abstractions and applications, pairs and
projections, disjoints union and definition by cases, \dots).
No new construction is needed, only the type system is new.

One striking feature of this expression of set theory in Deduction modulo
is the presence of the extensionality axiom. Extensionality axioms are
usually difficult to transform into computation rules. For instance,
for extensional simple type theory there is, as far as we know, no
expression in Deduction modulo and no normalization proof. The idea
is to define equality in such a way that it is extensional and then
prove that it is substitutive on the considered part of the language.
Whether this method can be generalized to extensional simple type
theory still remains to be investigated.

Our investigation on normalization has lead us to consider an
extension of Zermelo set theory with Strong Extensionality and
Transitive Closure. This raises the question of the interest
\emph{per se} of this theory. In particular, the fact that transitive
closure cannot always be constructed in Zermelo set
theory~\cite{EH99} can be seen as a weakness of this theory,
which is repaired by the transitive closure axiom.
However, we leave open the question of the various
axioms of set theory that can be added or removed from our choice of
axioms: both for weaker theories, for instance without the Transitive
Closure axiom and for stronger theories, for instance with the
collection scheme, the axiom of choice or the continuum hypothesis. 

Finally, the fact that set operations need to be decomposed into
more atomic operations raises the question of the relevance of the
choice of the notion of set for the foundation of mathematics. It might
be the case that founding mathematics directly on the notion of graph
would be more convenient.

\nocite{Acz99,Fri73,Kri98,Kri01,Kri03,McCPhD,Mel97,Myh73,Wer97}
\bibliographystyle{plain}
\bibliography{zerbiblio}

\end{document}